\documentclass{article}
\usepackage[utf8]{inputenc}
\usepackage{mathtools,tikz,pgfplots,xcolor,float,amssymb,amsthm}

\usetikzlibrary{math}
\pgfplotsset{compat=1.12}
\parskip=\medskipamount

\usetikzlibrary{patterns}

\usetikzlibrary{external}
\definecolor{myred}{RGB}{214,39,40}
\definecolor{mygreen}{RGB}{44,160,44}
\definecolor{myblue}{RGB}{31,119,180}
\definecolor{myorange}{RGB}{254,127,14}

\newlength{\hatchspread}
\newlength{\hatchthickness}
\newlength{\hatchshift}
\newcommand{\hatchcolor}{}
\tikzset{hatchspread/.code={\setlength{\hatchspread}{#1}},
         hatchthickness/.code={\setlength{\hatchthickness}{#1}},
         hatchshift/.code={\setlength{\hatchshift}{#1}},
         hatchcolor/.code={\renewcommand{\hatchcolor}{#1}}}
\tikzset{hatchspread=3pt,
         hatchthickness=0.4pt,
         hatchshift=0pt,
         hatchcolor=black}
\pgfdeclarepatternformonly[\hatchspread,\hatchthickness,\hatchshift,\hatchcolor]
   {custom north west lines}
   {\pgfqpoint{\dimexpr-2\hatchthickness}{\dimexpr-2\hatchthickness}}
   {\pgfqpoint{\dimexpr\hatchspread+2\hatchthickness}{\dimexpr\hatchspread+2\hatchthickness}}
   {\pgfqpoint{\dimexpr\hatchspread}{\dimexpr\hatchspread}}
   {
    \pgfsetlinewidth{\hatchthickness}
    \pgfpathmoveto{\pgfqpoint{0pt}{\dimexpr\hatchspread+\hatchshift}}
    \pgfpathlineto{\pgfqpoint{\dimexpr\hatchspread+0.15pt+\hatchshift}{-0.15pt}}
    \ifdim \hatchshift > 0pt
      \pgfpathmoveto{\pgfqpoint{0pt}{\hatchshift}}
      \pgfpathlineto{\pgfqpoint{\dimexpr0.15pt+\hatchshift}{-0.15pt}}
    \fi
    \pgfsetstrokecolor{\hatchcolor}
    \pgfusepath{stroke}
   }

\pgfdeclarepatternformonly[\hatchspread,\hatchthickness,\hatchshift,\hatchcolor]
   {custom north east lines}
   {\pgfqpoint{\dimexpr-2\hatchthickness}{\dimexpr-2\hatchthickness}}
   {\pgfqpoint{\dimexpr\hatchspread+2\hatchthickness}{\dimexpr\hatchspread+2\hatchthickness}}
   {\pgfqpoint{\dimexpr\hatchspread}{\dimexpr\hatchspread}}
   {
    \pgfsetlinewidth{\hatchthickness}
    \pgfpathmoveto{\pgfqpoint{\dimexpr\hatchshift-0.15pt}{-0.15pt}}
    \pgfpathlineto{\pgfqpoint{\dimexpr\hatchspread+0.15pt}{\dimexpr\hatchspread-\hatchshift+0.15pt}}
    \ifdim \hatchshift > 0pt
      \pgfpathmoveto{\pgfqpoint{-0.15pt}{\dimexpr\hatchspread-\hatchshift-0.15pt}}
      \pgfpathlineto{\pgfqpoint{\dimexpr\hatchshift+0.15pt}{\dimexpr\hatchspread+0.15pt}}
    \fi
    \pgfsetstrokecolor{\hatchcolor}
    \pgfusepath{stroke}
   }

\title{Classification of Real Solutions of the Fourth Painlev\'e Equation}
\author{Jeremy Schiff and Michael Twiton  \\
  Department of Mathematics, \\
  Bar-Ilan University, Ramat Gan, 5290002, Israel \\
   {\tt schiff@math.biu.ac.il} , {\tt mtwito101@gmail.com}}
\date{\today}

\newtheorem*{theorem*}{Theorem}

\begin{document}

\maketitle

\begin{abstract}
Painlev\'e transcendents are usually considered as complex functions
of a complex variable, but in applications it is often the real cases
that are of interest. Under a reasonable assumption (concerning the behavior
of a dynamical system associated with Painlev\'e IV, as discussed in a 
recent paper), we give a  number of results towards a  classification
of the real solutions of Painlev\'e IV (or, more precisely,
symmetric Painlev\'e IV), according to their asymptotic behavior
and singularities. We show the existence of globally nonsingular real 
solutions of symmetric Painlev\'e IV for arbitrary nonzero values of the parameters, with the dimension of the space of such solutions and their 
asymptotics depending on the signs of the parameters. We show that 
for a generic choice of the parameters, there exists a unique  finite
sequence of singularities for which symmetric Painlev\'e IV has a two-parameter family of solutions with this singularity sequence. There also exist 
solutions with singly infinite and doubly infinite sequences of singularities, and we identify which such sequences are possible (assuming generic 
asymptotics in the case of a singly infinite sequence). Most (but not all) 
of the special solutions of Painlev\'e IV correspond to nongeneric values 
of the parameters, but we mention some results for these solutions as well. 
\end{abstract}

\section{Introduction and Contents of This Paper}
The six Painlev\'e equations were initially discovered in the context of the classification of second order ordinary
differential equations with the property that the only movable singularities of their solutions are poles. In this context
solutions of the Painlev\'e equations are naturally considered as complex functions of a complex variable. Since their
initial discovery, however, many applications of Painlev\'e equations have emerged (see \cite{clarkson2006painleve, NIST:DLMF} for a comprehensive list).
In many of these applications, the relevant solutions are real functions of a real variable. It is therefore
of interest to have a classification of real solutions. In a recent paper \cite{SchiffTwiton} we described a dynamical systems approach
to the fourth Painlev\'e equation ($P_\mathrm{IV}$). In the current paper we use this approach
to develop a qualitative classification of the real solutions of 
$P_\mathrm{IV}$ for suitable parameter values. 

In fact we work with $sP_\mathrm{IV}$, the symmetric version of $P_\mathrm{IV}$. 
Recall that $P_\mathrm{IV}$ is the equation 
\begin{equation}
		\label{eq:p4}
		\frac{{\mathrm{d}}^{2}w}{{\mathrm{d}z}^{2}}
		=\frac{1}{2w}\left(\frac{\mathrm{d}w}{\mathrm{d}z}\right)^{2}
		+\frac{3}{2}w^{3}+4zw^{2}+2(z^{2}-\alpha)w+\frac{
			\beta}{w},
\end{equation}
with two parameters, $\alpha$ and $\beta$ (we take $\beta \le 0$).
$sP_\mathrm{IV}$ is the three-dimensional system
	\begin{subequations}
		\label{eq:fsys}
		\begin{align}
			\frac{\mathrm{d} f_1}{\mathrm{d} x}&=f_1 (f_2-f_3)+\alpha_1 \ , \\
			\frac{\mathrm{d} f_2}{\mathrm{d} x}&=f_2 (f_3-f_1)+\alpha_2 \ , \\
			\frac{\mathrm{d} f_3}{\mathrm{d} x}&=f_3 (f_1-f_2)+\alpha_3 \ , 
		\end{align}
	\end{subequations}
	subject to 
	\begin{equation}
		\label{eq:alphasum}
		\alpha_1 + \alpha_2 + \alpha_3 = 1 \ ,
	\end{equation}
        and
	\begin{equation}
		\label{eq:fsum}
		f_1+f_2+f_3=x \ .
	\end{equation}
$sP_\mathrm{IV}$ was known to Bureau \cite{bureau}  but was rediscovered by Adler \cite{adler} 
and Noumi and Yamada \cite{noumi1998affine,noumi1999symmetries}, amongst others.
The relationship between  $P_\mathrm{IV}$ and  $sP_\mathrm{IV}$ is a little subtle: 
If $f_1,f_2,f_3$ is a solution of \eqref{eq:fsys}--\eqref{eq:fsum} and   we set $w(z) = -\sqrt{2} f_1(x)$,
where $z = \frac{x}{\sqrt{2}}$, then $w(z)$ is a solution of \eqref{eq:p4} with parameter values
$\alpha = \alpha_3 - \alpha_2$ and $\beta = -2 \alpha_1^2$.  
But also, by the evident cyclic symmetry of  $sP_\mathrm{IV}$, 
if we set $w(z) = -\sqrt{2} f_2(x)$, then $w(z)$ is a solution of \eqref{eq:p4} with parameter values
$\alpha = \alpha_1 - \alpha_3$ and $\beta = -2 \alpha_2^2$,    
and if we set $w(z) = -\sqrt{2} f_3(x)$, then $w(z)$ is a solution of \eqref{eq:p4} with parameter values
$\alpha = \alpha_2 - \alpha_1$ and $\beta = -2 \alpha_3^2$. Thus the three components of a solution of the
$sP_\mathrm{IV}$  system
generically give three distinct solutions of  $P_\mathrm{IV}$, with different parameter values. Going the other way 
(i.e. using a solution of $P_\mathrm{IV}$  to find  a solution  of  $sP_\mathrm{IV}$) is more involved, as 
if $\beta = - 2 \alpha_1^2$ then $\alpha_1 = \pm \sqrt{ -\frac{\beta}{2} }$, so each of the maps from a solution of 
$sP_\mathrm{IV}$ to a solution of  $P_\mathrm{IV}$ can be inverted in two different ways. 
These ambiguities give rise to some of the symmetries or B\"acklund transformations of
$P_\mathrm{IV}$ and $sP_\mathrm{IV}$ \cite{clarkson2008fourth}; we will describe these more succinctly in Section 2.  
From here on we work with $sP_\mathrm{IV}$ throughout; all results can be translated to equivalent results for
$P_\mathrm{IV}$ with $\beta\le 0$, but these are less elegant.

The parameters of $sP_\mathrm{IV}$ are $\alpha_1,\alpha_2,\alpha_3$
satisfying (\ref{eq:alphasum}). If we introduce $\xi,\eta$  via
\begin{eqnarray}
    \alpha_1 &=&  \frac13 + \xi \ ,  \\
    \alpha_2 &=&  \frac13 -  \frac12 \xi  + \frac{\sqrt{3}}{2}\eta \ ,  \\
    \alpha_3 &=&  \frac13 -  \frac12 \xi  - \frac{\sqrt{3}}{2}\eta \ ,  
\end{eqnarray}
then a set of parameters corresponds to a point on the $(\xi,\eta)$ plane. The lines on which one of 
$\alpha_1,\alpha_2,\alpha_3$ is an integer form a triangular lattice on this plane, see Figure \ref{fig:lattice}. By generic parameter values we mean
any value of the parameters for which all the $\alpha_i$ are noninteger. For values of the parameters corresponding to
the vertices of the lattice or the centers of the cells of the lattice
(marked, respectively, by black and white circles in Figure \ref{fig:lattice}), there exist
rational solutions, see \cite{clarkson2008fourth}. For other nongeneric values of the parameters, solutions are known in
terms of special functions \cite{clarkson2008fourth}. 
Section 5 of this paper will be devoted to special solutions; 
apart from this, the focus of this paper is on generic parameter values. 

\begin{figure}
  \begin{center}
      \includegraphics[width=0.5\textwidth]{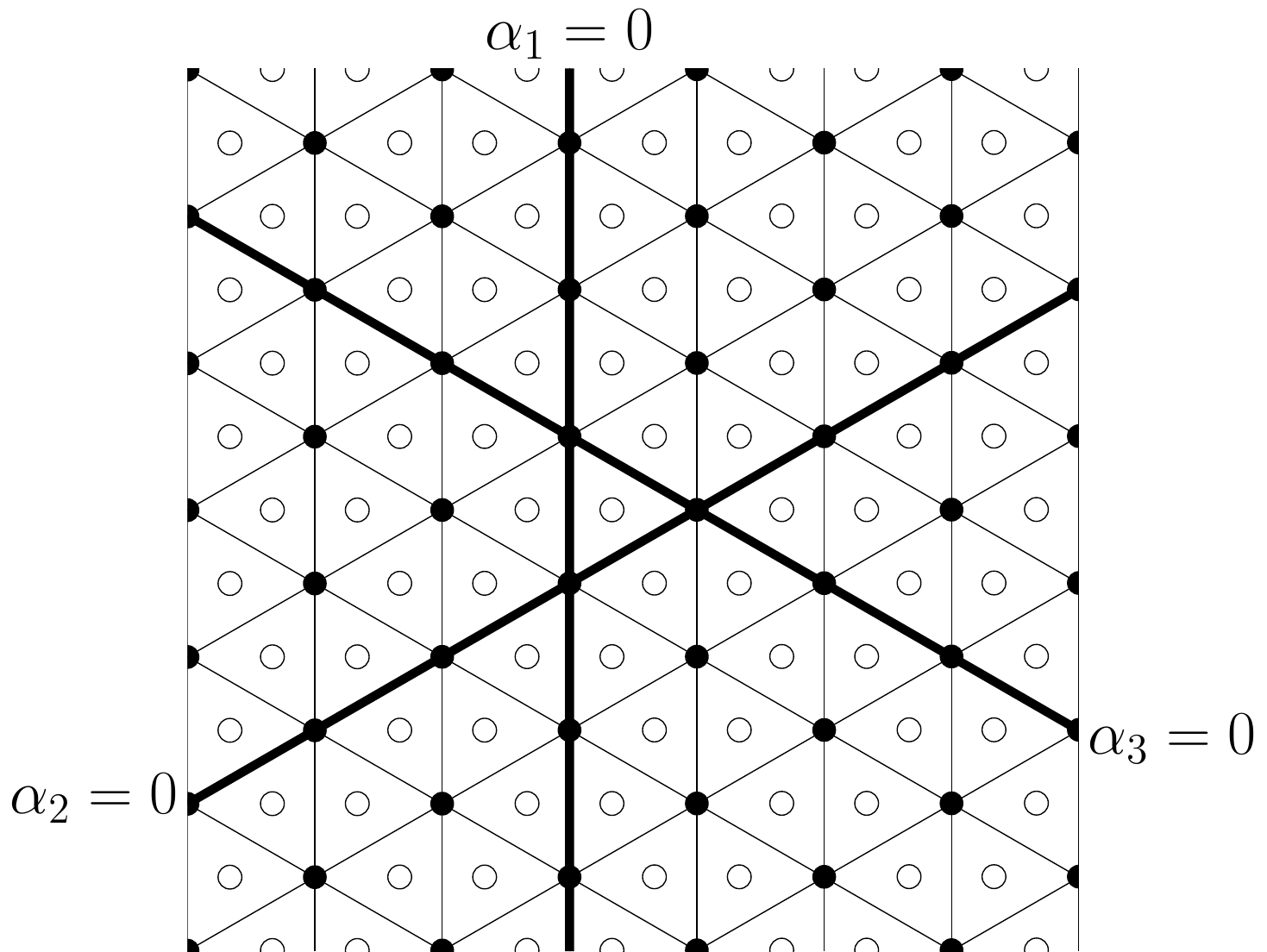}
  \end{center}
  \caption{
  The plane $\alpha_1+\alpha_2+\alpha_3=1$ viewed from the positive normal direction, with lines drawn where one of the $\alpha_i$ is integral.
  The lines on which one the $\alpha_i$ vanish are marked in bold.  
  The values of $\alpha_1,\alpha_2,\alpha_3$ at a point are given by
  (suitable signed) perpendicular distances to these lines.  }
    \label{fig:lattice}
 \end{figure}
 
In Section 2 of this paper we review the necessary results from \cite{SchiffTwiton}. In \cite{SchiffTwiton} we considered the
Poincar\'e compactification of $sP_\mathrm{IV}$, showed it had $14$ fixed points on its boundary, and studied the stability of these
fixed points. The starting point of the current paper will be the assumption that all orbits in the interior of the compactification both
``start from'' and ``end at'' one of these fixed points (more formally: each orbit has a single $\alpha-$limit point and a single $\omega-$limit
point). This is currently an assumption, as in \cite{SchiffTwiton} we could not exclude the possibility of orbits with $\alpha-$ or $\omega-$
limit sets given by certain periodic orbits on the boundary. The assumption is consistent with everything known, and, in particular, with a
substantial amount of numerical evidence. 

Also in Section 2 we present the symmetry group of $sP_\mathrm{IV}$. This was not considered in \cite{SchiffTwiton}, and the use
of the symmetry group allows us to significantly  extend the results of \cite{SchiffTwiton}. 

In Section 3 we derive results concerning real solutions of $sP_\mathrm{IV}$ with no singularities on the  entire real axis. These correspond
directly to solutions of the compactification of $sP_\mathrm{IV}$ going from one of the 4 fixed points we label $B_1^-,B_2^-,B_3^-,C^-$
as the independent variable tends to
$-\infty$ to one of the 4 fixed points we label $B_1^+,B_2^+,B_3^+,C^+$ as 
the independent variable tends to 
$+\infty$.    We show the following:
\begin{itemize}
\item   If $\alpha_1,\alpha_2,\alpha_3$ are all positive, then there exist the following  families of solutions with no singularities on the
  entire real axis:
  \begin{itemize}
  \item  a two-parameter family going from $C^-$ to $C^+$.  
  \item  one-parameter families going from $C^-$ to each of $B_1^+,B_2^+,B_3^+$. 
  \item  one-parameter families going from each of $B_1^-,B_2^-,B_3^-$ to $C^+$. 
  \item  isolated solutions going from $B_i^-$ to $B_j^+$ for $i\not= j \in\{1,2,3\}$.   
   \end{itemize}
\item If one of $\alpha_1,\alpha_2,\alpha_3$ is negative, say $\alpha_k$, and the other two positive, then there exist the following  families of solutions with no singularities on the entire real axis:
  \begin{itemize}
  \item  a one parameter family going from $C^-$ to $B_{k-1~{\rm mod}~3}^+$. 
  \item  a one parameter family going from $B_{k-1~{\rm mod}~3}^-$  to $C^+$. 
  \item  isolated solutions going from $B_{k-1~{\rm mod}~3}^-$    to each of $B_k^+, B_{k+1~{\rm mod}~3}^+$ 
          and from each of  $B_k^-, B_{k+1~{\rm mod}~3}^-$   to   $B_{k-1~{\rm mod}~3}^+$. 
   \end{itemize} 
\item If one of $\alpha_1,\alpha_2,\alpha_3$ is positive, say $\alpha_k$, and the other two negative, then there exist  at least two solutions with no singularities on the entire real axis, 
    going from $B_k^-$ to $B_{k+1~{\rm mod}~3}^+$
    and from  $B_{k+1~{\rm mod}~3}^-$ to $B_k^+$. 
\end{itemize}

In Section 4 we move on to solutions of $sP_\mathrm{IV}$ with singularities. Solutions of $sP_\mathrm{IV}$ have 3 different types of singularities, at each of which two of the functions $f_1,f_2,f_3$
have simple poles and the third has a zero.  Unlike the global solutions described above, these do not correspond to a single
solution of the compactified system, but rather a sequence of solutions. The first (second) ((third)) type of singularity corresponds to a pair of solutions of the compactification,
with the first ``ending'' at $A_1^+$ ($A_2^+$) (($A_3^+$))  and the second ``starting'' at $A_1^-$ ($A_2^-$) (($A_3^-$)). Here $A_1^+, A_2^+, A_3^+,A_1^-, A_2^-, A_3^-$ denote the $6$
fixed points of the compactification not yet mentioned.  Using restrictions on orbits of the compactification shown in  \cite{SchiffTwiton}, and further restrictions
found using the symmetry group, we give a full description of the possible sequences of singularities for generic solutions of 
$sP_\mathrm{IV}$ for generic values of the parameters. In particular we prove the following: {\em for any generic choice of the parameters, there exists a unique {\em finite}
  sequence of singularities for which $sP_\mathrm{IV}$ has a two-parameter family of solutions with this singularity sequence}. In the previous paragraph
we already stated that if all the $\alpha_i$ are positive, there exists a two-parameter family of solutions with no singularities; now we can add that
if all the $\alpha_i$ are positive, then there do not exist two-parameter families of solutions with a nonzero finite number of singularities. Similarly, for
cases where the $\alpha_i$ are of mixed sign, we have seen that there does not exist a two-parameter family of solutions with no singularities; but there does exist
a two-parameter family of solutions with a specific finite singularity sequence.  In addition to describing the solutions with finite sequences of singularities, we identify
which  singly infinite and doubly infinite sequences of singularities are allowed. 

Section 5 discusses special solutions.  Our main intention in this section is to
briefly examine the singularities of the special solutions of $sP_\mathrm{IV}$ 
on the real line, but in addition we mention two other results that we believe are new  (or, at least, generalizations
of existing results). In the case of rational solutions, we show 
how to find a pair of polynomial equations that 
the functions  $f_1,f_2,f_3$ satisfy; the rational solutions can be obtained by solving these polynomial equations, along with constraint (\ref{eq:fsum}).  
For the case of nongeneric value of the parameters, i.e. parameters for which one of the $\alpha_i$ is an integer, we show that 
there are special solutions obtained from the solution of  a single first order differential equation. This is a generalization of a classical result (see for example \cite{FA}) 
that for the case $\beta=-2(\alpha\pm 1)^2$, there are special solutions of $P_\mathrm{IV}$ that can be obtained from the solution of  a Riccati equation.  

Section 6 contains a brief summary and some concluding remarks. 
In Appendix A we briefly describe the numerical methods used to integrate $sP_\mathrm{IV}$ through singularities. In Appendix B we briefly present some numerical results on global $B$ to $B$ type solutions, the significance of which is explained in Section 3.

\section{Background}

\subsection{The Poincar\'e Compactification of $sP_\mathrm{IV}$}

In \cite{SchiffTwiton} we described the Poincar\'e compactification of  $sP_\mathrm{IV}$ (which is also related to the compactification of $P_\mathrm{IV}$
on a projective space described by Chiba \cite{chiba}). The Poincar\'e compactification is a flow on the closed unit ball in ${\bf R}^3$.
Solutions of $sP_\mathrm{IV}$ between singularities correspond to orbits of the compactification in the interior of the ball; the boundary
(which we call ``the sphere at infinity'') is an invariant submanifold, on which the flow can be completely solved. All orbits in the interior of the
ball have $\alpha$--limit sets on the closed lower hemisphere of sphere at infinity and $\omega$--limit sets on the closed upper hemisphere. 
The flow has $14$ fixed points on the sphere at infinity, four (which we label $B_1^-,B_2^-,B_3^-,C^-$) in the open lower hemisphere, 
four (which we label $B_1^+,B_2^+,B_3^+,C^+$) in the open upper hemisphere, and six (which we label $A_1^-,A_2^-,A_3^-, A_1^+,A_2^+,A_3^+$) on the equator.
The points $A_1^+,A_2^+,A_3^+,C^+$ are asymptotically stable, in the sense that any orbit in the interior of the ball, sufficiently close to one of these
points, will converge to that point as $t\rightarrow + \infty$ ($t$ denotes the independent variable of the compactification). Similarly,  the points $A_1^-,A_2^-,A_3^-,C^-$ are asymptotically unstable (i.e. stable
as $t\rightarrow -\infty$).  The points $B_1^-, B_2^-, B_3^-, B_1^+, B_2^+, B_3^+$ are of mixed stability: There are two-dimensional stable manifolds
in the interior of the ball associated with each of the points $B_1^+,B_2^+,B_3^+$ (and one-dimensional unstable manifolds on the sphere at infinity).
Similarly, there are two dimensional unstable manifolds
in the interior of the ball associated with each of the points $B_1^-,B_2^-,B_3^-$.

In \cite{SchiffTwiton} we could not exclude the possibility of there being orbits in the interior of the ball with $\alpha-$ or $\omega-$limit sets
that are closed orbits on the sphere at infinity (and not one of the $14$ fixed points). But we have no numerical evidence for such orbits, and
neither is there 
any suggestion in the extensive literature on $P_\mathrm{IV}$ of a solution with appropriate asymptotic behavior. Therefore {\em we proceed in this paper on the assumption
that no such orbit exists}. We then have two partitions of the interior of the ball. The first is into the four open sets that are the basins of attraction of 
each of the points $A_1^+,A_2^+,A_3^+,C^+$ as $t\rightarrow+\infty$, separated by the three nonintersecting stable manifolds of the points $B_1^+,B_2^+,B_3^+$.  
The second is into the four open sets that are the ``basins of repulsion'' of  $A_1^-,A_2^-,A_3^-,C^-$ as $t\rightarrow-\infty$, separated by the three nonintersecting unstable manifolds of the points $B_1^-,B_2^-,B_3^-$.  
Note that $sP_\mathrm{IV}$ has the obvious symmetry $f(x)\rightarrow -f(-x)$, which  relates the two partitions. 

In addition to performing local analysis of the fixed points, in \cite{SchiffTwiton} we considered the question of whether there could exist orbits connecting
each of the four fixed points $A_1^-,A_2^-,A_3^-,C^-$ to each of the four fixed points  $A_1^+,A_2^+,A_3^+,C^+$, and gave a set of rules for the permitted transitions, deduced
by looking at the signs of $f_1,f_2,f_3$ near the fixed points, and using the fact that the signs of the parameters $\alpha_1,\alpha_2,\alpha_3$ determine the changes of
sign of $f_1,f_2,f_3$ respectively at their zeros between singularities. We reproduce the rules of permitted and forbidden transitions in Table \ref{tab:excls}.
In addition to showing forbidden transitions (indicated
with an {\sffamily X}), for  permitted transitions we give a list of numbers, showing which of the functions $f_1,f_2,f_3$ change sign in the course of the transition.

\begin{table}[h]
		\begin{center}
			\begin{tabular}{c|cccc}
				$+++$   & $C^+$  & $A_1^+$ & $A_2^+$ &  $A_3^+$  \\ \hline
				$C^-$   &  1,2,3 &  1,3          & 1,2     & 2,3  \\
				$A_1^-$ &    1,3    &  \text{\sffamily X} & 1 & 3 \\
				$A_2^-$ &   1,2     &   1         & \text{\sffamily X}  & 2 \\ 
				$A_3^-$ &   2,3     & 3           & 2  &  \text{\sffamily X}
			\end{tabular}
		\end{center}
		\begin{center}
			\begin{tabular}{c|cccc}
				$++-$    & $C^+$  & $A_1^+$ & $A_2^+$ &  $A_3^+$  \\\hline
				$C^-$   & \text{\sffamily X} & \text{\sffamily X}       &  1,2     & 2         \\
				$A_1^-$ & \text{\sffamily X} & \text{\sffamily X}       &  1     &         \\
				$A_2^-$ &  1,2       &  1     &  1,2,3     &  2,3        \\ 
				$A_3^-$ &   2     &      &  2,3     &  \text{\sffamily X}
			\end{tabular}
			~~
			\begin{tabular}{c|cccc}
				$-++$    & $C^+$  & $A_1^+$ & $A_2^+$ &  $A_3^+$  \\ \hline
				$C^-$   & \text{\sffamily X} &  3     & \text{\sffamily X}     &  2,3       \\
				$A_1^-$ &    3     & \text{\sffamily X}       &      &  1,3       \\
				$A_2^-$ & \text{\sffamily X}       &       &  \text{\sffamily X}      &      2   \\ 
				$A_3^-$ &     2,3   & 1,3      &    2   & 1,2,3 
			\end{tabular}
			~~
			\begin{tabular}{c|cccc}
			        $+-+$   & $C^+$  & $A_1^+$ & $A_2^+$ &  $A_3^+$  \\ \hline
				$C^-$   & \text{\sffamily X} &  1,3    &    1   & \text{\sffamily X}       \\
				$A_1^-$ &   1,3    &   1,2,3     &  1,2     &   3      \\
				$A_2^-$ &    1    &  1,2     &  \text{\sffamily X}      &     \\ 
				$A_3^-$ & \text{\sffamily X}       &    3   &     & \text{\sffamily X}
			\end{tabular}
		\end{center}
		\begin{center}
			\begin{tabular}{c|cccc}
				$--+$    & $C^+$     & $A_1^+$ & $A_2^+$  &  $A_3^+$  \\ \hline
				$C^-$   & \text{\sffamily X}  &     3    &\text{\sffamily X} & \text{\sffamily X} \\
				$A_1^-$ &   3        &    1,2,3     &    2     &  3         \\
				$A_2^-$ & \text{\sffamily X}  & 2        &\text{\sffamily X} &          \\ 
				$A_3^-$ & \text{\sffamily X}  & 3        &   &  \text{\sffamily X} 
			\end{tabular}
			~~
			\begin{tabular}{c|cccc}
				$+--$   & $C^+$  & $A_1^+$ & $A_2^+$ &  $A_3^+$  \\ \hline
				$C^-$   & \text{\sffamily X}  & \text{\sffamily X} &    1       &  \text{\sffamily X}  \\
				$A_1^-$ & \text{\sffamily X}  & \text{\sffamily X} &    1       &           \\
				$A_2^-$ &   1        &  1        &  1,2,3         & 3           \\ 
				$A_3^-$ & \text{\sffamily X}  &        &   3        &  \text{\sffamily X} 
			\end{tabular}
			~~
			\begin{tabular}{c|cccc}
				$-+-$   & $C^+$  & $A_1^+$ & $A_2^+$ &  $A_3^+$  \\ \hline
				$C^-$   & \text{\sffamily X}&  \text{\sffamily X}  & \text{\sffamily X} & 2         \\
				$A_1^-$ & \text{\sffamily X} & \text{\sffamily X}&  & 1    \\
				$A_2^-$ & \text{\sffamily X} &         & \text{\sffamily X} & 2         \\ 
				$A_3^-$ &  2       &    1     &     2    & 1,2,3
			\end{tabular}
		\end{center}
	\caption{
          Excluded transitions:
           Top line: $\alpha_1,\alpha_2,\alpha_3>0$.
           Middle line:  
			$\alpha_1,\alpha_2>0$, $\alpha_3<0$ (left) ; 
			$\alpha_2,\alpha_3>0$, $\alpha_1<0$ (middle) ;
		       $\alpha_3,\alpha_1>0$, $\alpha_2<0$ (right).
           Bottom line: 
			$\alpha_1,\alpha_2<0$, $\alpha_3>0$ (left) ; 
			$\alpha_2,\alpha_3<0$, $\alpha_1>0$ (middle) ;
	   $\alpha_3,\alpha_1<0$, $\alpha_2>0$ (right). {\sffamily X} indicates an excluded transition. 
	   For  permitted transitions a list of numbers is given, showing which of the functions $f_1,f_2,f_3$ change sign in the course of the transition. }
                        \label{tab:excls}
\end{table}

\subsection{The Symmetry Group of $sP_\mathrm{IV}$}

$sP_\mathrm{IV}$ clearly has a ${\bf Z}_3$  cyclic symmetry generated by the transformation $\sigma$, where 
$$   \sigma(\alpha_i) = \alpha_{i+1~{\rm mod}~3}   \ , \qquad\sigma(f_i) = f_{i+1~{\rm mod}~3}  \ . $$
It is straightforward to directly verify that there is a further symmetry  $\tau$  given by
$$
\begin{array}{lll}
  \tau(\alpha_1) = -\alpha_1\ ,   &  \tau(\alpha_2) = \alpha_2 + \alpha_1\ ,   &  \tau(\alpha_1) = \alpha_3 +\alpha_1\ ,   \\
  \tau(f_1) = f_1 \ ,  &  \tau(f_2) = f_2 + \frac{\alpha_1}{f_1} \ ,   &  \tau(f_3) = f_3 - \frac{\alpha_1}{f_1}   \ .  
\end{array}   
$$
(Other authors prefer to introduce three further symmetries 
$\tau_1=\tau$, $\tau_2=\sigma\tau\sigma^2$, $\tau_3=\sigma^2\tau\sigma$.)
The generators $\sigma$ and $\tau$ satisfy the relations
$$  \sigma^3 = \tau^2 = (\tau \sigma \tau \sigma^2)^3 = I  \ .  $$ 
The infinite group generated by $\sigma$ and $\tau$ is known as the extended affine Weyl group of type $A_2^{(1)}$ \cite{noumi1998affine,noumi1999symmetries}. 
The action of the group on the space of parameters generates the entire parameter space from a single triangular cell in Figure 1. Thus, in principle it
suffices to know the solutions of  $sP_\mathrm{IV}$ just in the case, say,  that all the $\alpha_i$ are non-negative. However, since the transformation $\tau$
can add or remove singularities, it is still important to understand the qualitative behavior of solutions for all parameter values. 

\section{Global Solutions of $sP_\mathrm{IV}$}

From what we have written above concerning the Poincar\'e compactification of $sP_\mathrm{IV}$, it is clear that global solutions of $sP_\mathrm{IV}$,
with no singularities on the entire real axis, correspond to orbits of the compactification going from any of the points $B_1^-,B_2^-, B_3^-,C^-$ to
any of the points $B_1^+,B_2^+, B_3^+,C^+$.  From Table \ref{tab:excls} we see that transitions from $C^-$ to $C^+$ are only permitted in the case that $\alpha_1,\alpha_2,\alpha_3$
are all positive.  By considering the signs of the solutions near the various points (using equations (16)--(17) in 
\cite{SchiffTwiton}) it is straightforward to determine which transitions 
are permitted and which are prohibited between $B$ and $C$ type points. See 
Table \ref{tab:excls2}. 

\begin{table}
	\begin{center}
			\begin{tabular}{c|cccc}
				$+++$   & $C^+$  & $B_1^+$ & $B_2^+$ &  $B_3^+$  \\ \hline
				$C^-$   &  1,2,3 & 1,2  & 2,3  & 1,3  \\
				$B_1^-$ & 1,2   &  \text{\sffamily X} & 2  & 1 \\
				$B_2^-$ & 2,3  & 2  & \text{\sffamily X}  & 3 \\ 
				$B_3^-$ & 1,3  & 1  & 3  &  \text{\sffamily X}
			\end{tabular}
		\end{center}
		\begin{center}
	\begin{tabular}{c|cccc}
				$++-$   & $C^+$  & $B_1^+$ & $B_2^+$ &  $B_3^+$  \\ \hline
				$C^-$   & \text{\sffamily X}  & \text{\sffamily X}   & 2  & \text{\sffamily X}  \\
				$B_1^-$ &  \text{\sffamily X} & \text{\sffamily X}  &  2 &\text{\sffamily X}  \\
				$B_2^-$ &  2 & 2  & \text{\sffamily X}  &  \\ 
				$B_3^-$ &\text{\sffamily X}   & \text{\sffamily X}  &   &  \text{\sffamily X}
			\end{tabular}
			~~
				\begin{tabular}{c|cccc}
				$-++$   & $C^+$  & $B_1^+$ & $B_2^+$ &  $B_3^+$  \\ \hline
				$C^-$   &  \text{\sffamily X} & \text{\sffamily X}  & \text{\sffamily X}  & 3 \\
				$B_1^-$ & \text{\sffamily X}  & \text{\sffamily X}  &\text{\sffamily X}   &  \\
				$B_2^-$ &  \text{\sffamily X} & \text{\sffamily X}  & \text{\sffamily X}  & 3 \\ 
				$B_3^-$ & 3  &   & 3  &  \text{\sffamily X}
			\end{tabular}
		   ~~
		        \begin{tabular}{c|cccc}
				$+-+$   & $C^+$  & $B_1^+$ & $B_2^+$ &  $B_3^+$  \\ \hline
				$C^-$   & \text{\sffamily X}  & 1  & \text{\sffamily X}  & \text{\sffamily X} \\
				$B_1^-$ & 1  & \text{\sffamily X}  &   & 1 \\
				$B_2^-$ & \text{\sffamily X}  &   &  \text{\sffamily X} & \text{\sffamily X} \\ 
				$B_3^-$ & \text{\sffamily X}  & 1  & \text{\sffamily X}  &  \text{\sffamily X}
			\end{tabular}
		\end{center}
		\begin{center}
			\begin{tabular}{c|cccc}
				$--+$   & $C^+$  & $B_1^+$ & $B_2^+$ &  $B_3^+$  \\ \hline
				$C^-$   & \text{\sffamily X}  &\text{\sffamily X}   & \text{\sffamily X}  & \text{\sffamily X} \\
				$B_1^-$ & \text{\sffamily X}  & \text{\sffamily X}  &\text{\sffamily X}   &  \\
				$B_2^-$ & \text{\sffamily X}  & \text{\sffamily X}  &\text{\sffamily X}   &\text{\sffamily X}  \\ 
				$B_3^-$ & \text{\sffamily X}  &   & \text{\sffamily X}  & \text{\sffamily X} 
			\end{tabular}
			~~
				\begin{tabular}{c|cccc}
				$+--$   & $C^+$  & $B_1^+$ & $B_2^+$ &  $B_3^+$  \\ \hline
				$C^-$   &  \text{\sffamily X} & \text{\sffamily X}  & \text{\sffamily X}  &  \text{\sffamily X}\\
				$B_1^-$ &  \text{\sffamily X} & \text{\sffamily X}  &   & \text{\sffamily X} \\
				$B_2^-$ &  \text{\sffamily X} &   &  \text{\sffamily X} &\text{\sffamily X}  \\
				$B_3^-$ &  \text{\sffamily X} & \text{\sffamily X}  & \text{\sffamily X}  &  \text{\sffamily X}
			\end{tabular}
			~~
				\begin{tabular}{c|cccc}
				$-+-$   & $C^+$  & $B_1^+$ & $B_2^+$ &  $B_3^+$  \\ \hline
				$C^-$   &\text{\sffamily X}   &\text{\sffamily X}   &\text{\sffamily X}   & \text{\sffamily X} \\
				$B_1^-$ & \text{\sffamily X}  &\text{\sffamily X}   & \text{\sffamily X}  & \text{\sffamily X} \\
				$B_2^-$ &\text{\sffamily X}   & \text{\sffamily X}  & \text{\sffamily X}  &  \\ 
				$B_3^-$ &\text{\sffamily X}   &\text{\sffamily X}   &   & \text{\sffamily X} 
			\end{tabular}
		\end{center}
	\caption{
          Excluded transitions:
           Top line: $\alpha_1,\alpha_2,\alpha_3>0$.
           Middle line:  
			$\alpha_1,\alpha_2>0$, $\alpha_3<0$ (left) ; 
			$\alpha_2,\alpha_3>0$, $\alpha_1<0$ (middle) ;
		       $\alpha_3,\alpha_1>0$, $\alpha_2<0$ (right).
           Bottom line: 
			$\alpha_1,\alpha_2<0$, $\alpha_3>0$ (left) ; 
			$\alpha_2,\alpha_3<0$, $\alpha_1>0$ (middle) ;
	   $\alpha_3,\alpha_1<0$, $\alpha_2>0$ (right). {\sffamily X} indicates an excluded transition. 
	   For  permitted transitions a list of numbers is given, showing which of the functions $f_1,f_2,f_3$ change sign in the course of the transition.
	      } \label{tab:excls2}
\end{table}


To show the existence of all the permitted transitions, we consider the equatorial plane of the Poincar\'e compactification.
The stable manifolds of the points $B_1^+,B_2^+,B_3^+$ intersect the equatorial plane transversely, and can only reach the boundary at one of the $A^-$ points.
They divide the equatorial plane into regions of points on orbits tending to the four points $A_1^+,A_2^+,A_3^+,C^+$. In Figure \ref{fig:awesome} we show this
division (computed numerically) in three cases. In the $+++$ case, from Table 1, the neighborhood of each of the $A^-$ points can be divided into $3$ regions
and no more than $3$ regions, corresponding to the $3$ possible ``destinations'' of orbits starting at any of the $A^-$ points. Thus two of the three stable
manifolds must meet each of the $A^-$ points and we have the triangular configuration shown.  In the $++-$ case, there can be at most two regions in the neighborhood
of one of the $A^-$ points, at most three at one of the others, and possibly all four at the last. The configuration shown is clearly the only option.
In the $+--$ case, two points can have at most two regions in their neighborhood, and thus only one stable manifold can reach these points. The other point
can have up to four regions in its neighborhood, and thus one of the stable manifolds must form a loop to meet this point twice. Note that it is possible that the loop might be between the other two stable manifolds, or between one of them and the boundary, as in the third image in Figure \ref{fig:awesome}. 

\begin{figure}
  \includegraphics[width=5.9cm]{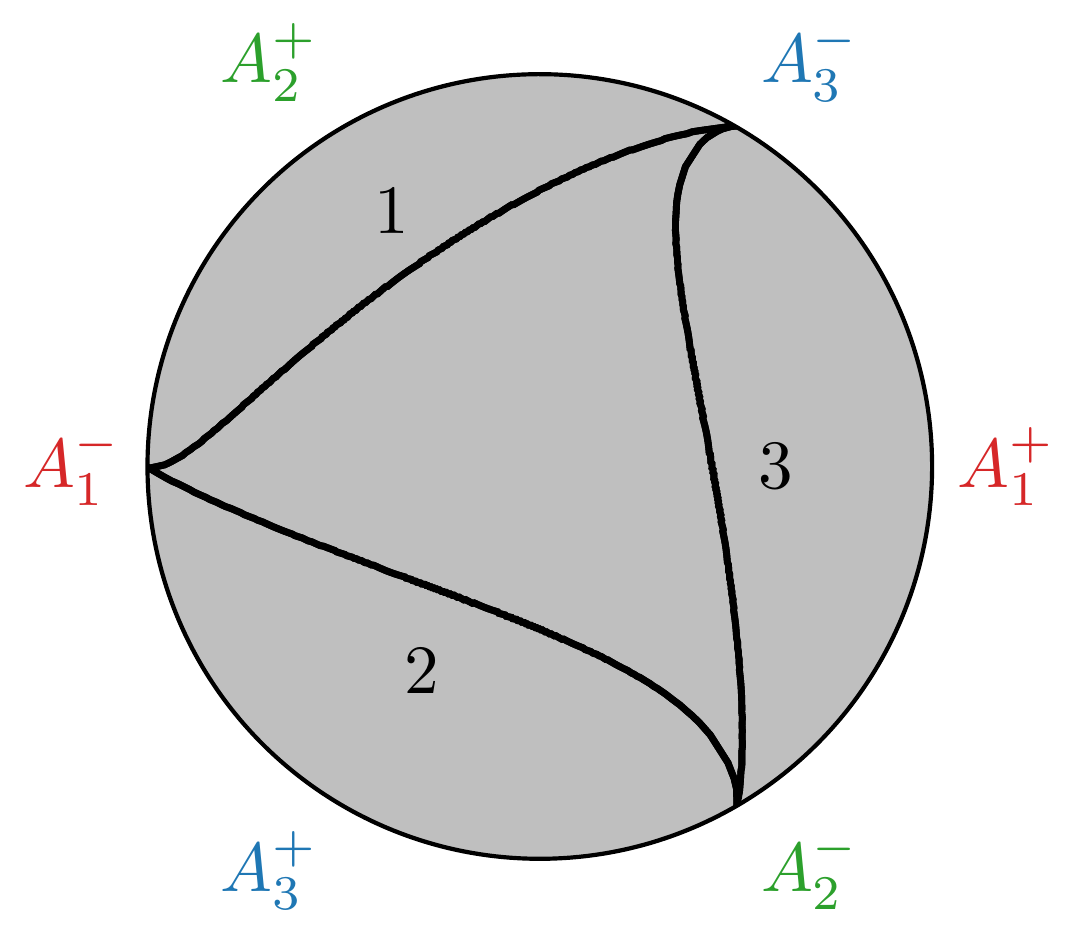}
  \includegraphics[width=5.9cm]{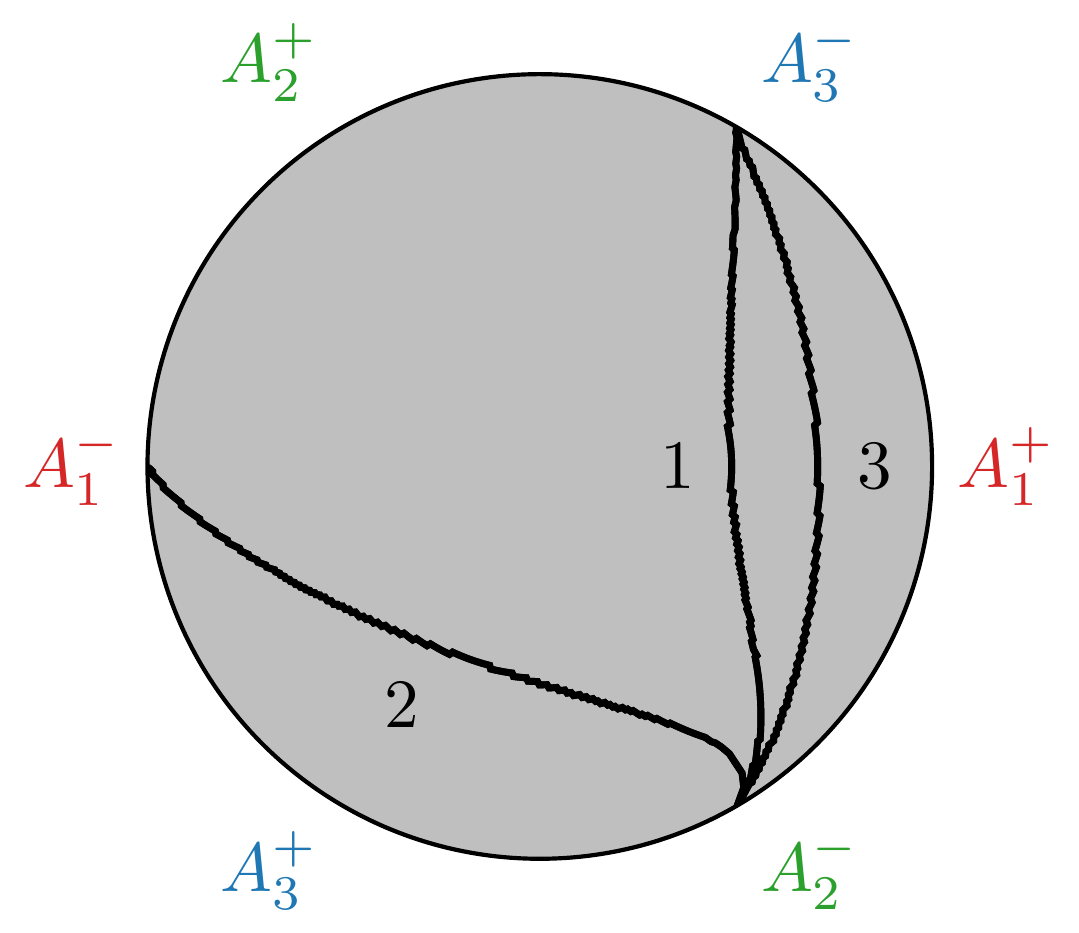}
   \includegraphics[width=5.9cm]{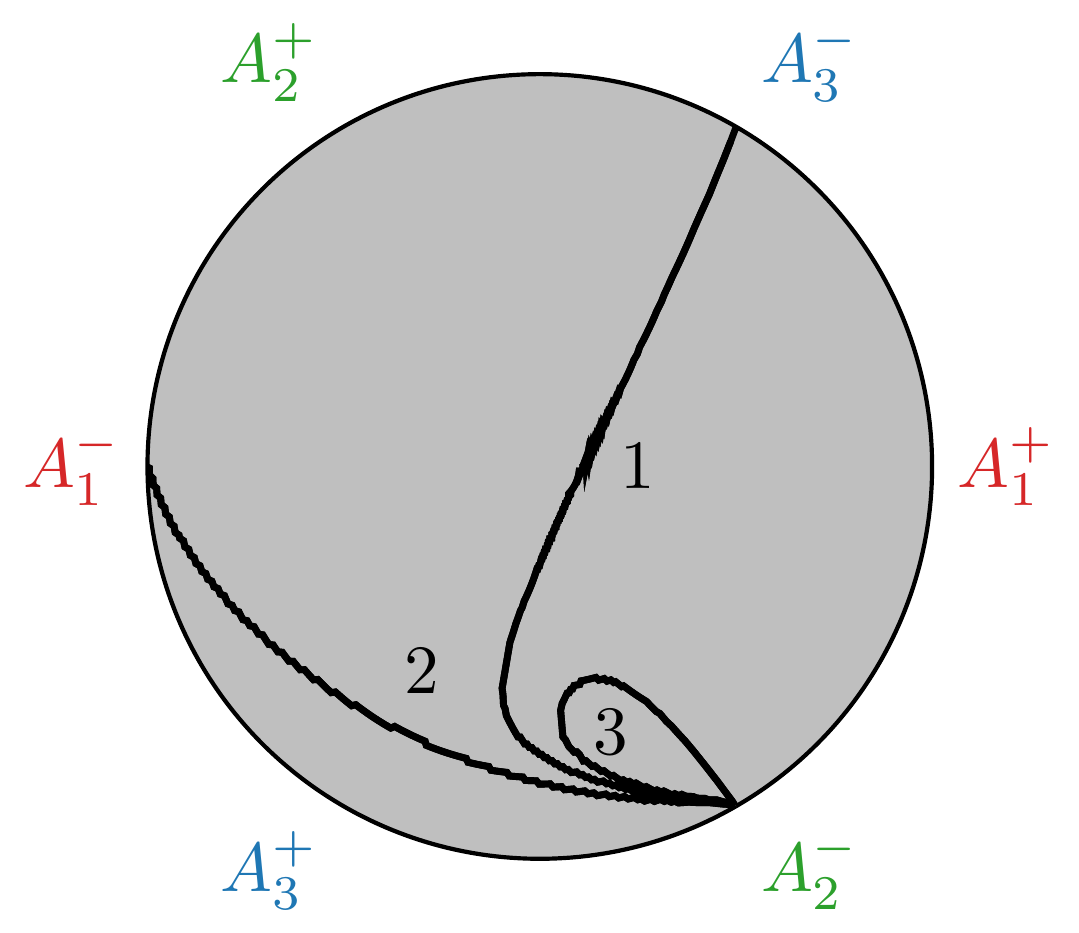}
  \caption{The equatorial plane is divided up by the stable manifolds of $B_1^+$,$B_2^+$,$B_3^+$  into basins of attraction of 
  $A_1^+,A_2^+,A_3^+,C^+$. Three possible divisions are shown, for the $+++$ case (left), the $++-$ case (middle) and the $+--$ case (right). The black curves
  are labelled to show which are the stable manifolds of $B_1^+,B_2^+,B_3^+$.  
  }  \label{fig:awesome} 
\end{figure}

\begin{figure}[t]
    \centering
    \includegraphics[width=0.6\textwidth]{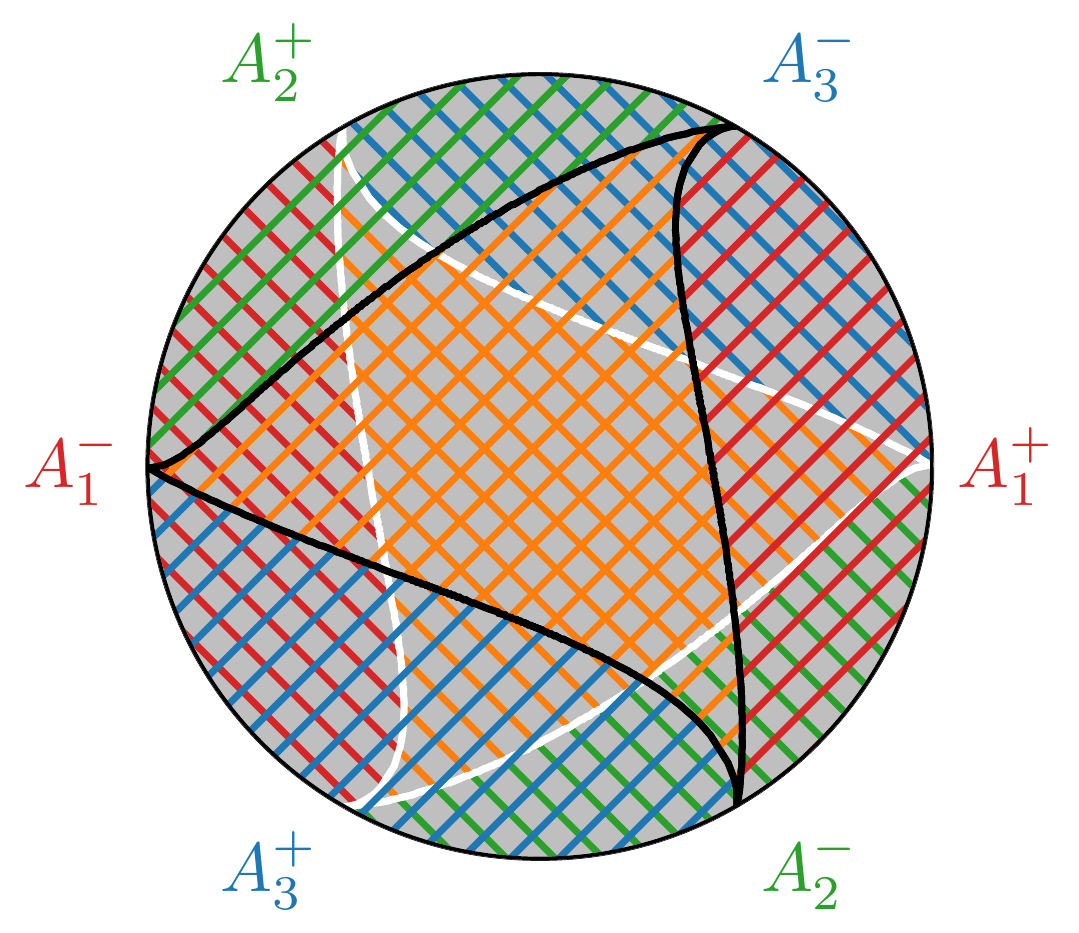}
    \caption{The equatorial plane with parameter values $\alpha_1 = 0.2, \alpha_2=0.3, \alpha_3 =0.5$ ($+++$ case). The black curves are the stable manifolds of the $B^+$ points, the white curves are the unstable manifolds of the $B^-$ points. The color of the positively-sloped hatching in a region shows the limit of orbits in this region as $t \rightarrow +\infty$: red denotes $A_1^+$, green denotes $A_2^+$, blue denotes $A_3^+$, orange denotes $C^+$. Similarly the color of the negatively-sloped hatching shows the limit as $t\rightarrow -\infty$.}
     \label{fig:ppp}
\end{figure}

To establish the existence of orbits going from  one of the $B^-$ or $C^-$ points to one of the $B^+$ or $C^+$ points we consider the division of the
equatorial plane by both the stable manifolds of the $B^+$ points and the unstable manifolds of the $B^-$ points, the latter being obtained by a
half turn from the former due to the $f(x) \rightarrow -f(-x)$ symmetry
of $sP_\mathrm{IV}$. See Figures  \ref{fig:ppp},\ref{fig:ppm},\ref{fig:pmm}
for the $+++$, $++-$ and $+--$ cases respectively.
In the $+++$ case it is clear that the rotated copy of each of the stable manifolds of the $B^+$ points (i.e. the unstable manifold of the corresponding $B^-$ point) must intersect the stable manifolds of the
other two $B^+$ points. It follows, by simple topological arguments, that in the $+++$ case there must be at least one open region in the disk corresponding to
solutions going from $C^-$ to $C^+$;
there must be at least one curve segment corresponding to solutions going from $C^-$ to any of the $B^+$ points and from any of the $B^-$ points to $C^+$; and
there must be at least one point corresponding to an orbit going from any of the points $B_i^-$ to the points $B_j^+$ with $j\not=i$.
In the $++-$ case we obtain (at least) two curve segments corresponding to solutions going from $C^-$ to one of the $B^+$ points (as allowed by the rules in Table 2), and
from one of the $B^-$ points to $C^+$. In addition, there are at least four $B$ to $B$ solutions. In the $+--$ case we obtain (at least) two $B$ to $B$ solutions.  
In short: solutions exhibiting every transition allowed by Table  \ref{tab:excls2} appear, with the expected number of parameters ($2$ for $C$ to $C$, $1$ for $C$ to $B$ and $B$ to $C$, and isolated solutions for $B$ to $B$). This is the result on nonsingular solutions described in the Introduction. The only assumption that has been made on the parameters in reaching the result is that none of the $\alpha_i$ vanish.  There is extensive discussion in the literature
\cite{bassom1992integral, bassom1993numerical, reeger2013painleve,  reeger2014painleve}
of real solutions of  
($P_\mathrm{IV}$) in the case $\beta=0$, which is precisely the case that we have excluded. However, there are clear similarities between the existing results in the case $\beta=0$ and our results, presumably reflecting the fact that some 
properties persist in an appropriate limit as one of the $\alpha_i$ tends to zero. 
\def\d{14pt} 

\begin{figure}[!h]
    \centering
    \includegraphics[width=0.6\textwidth]{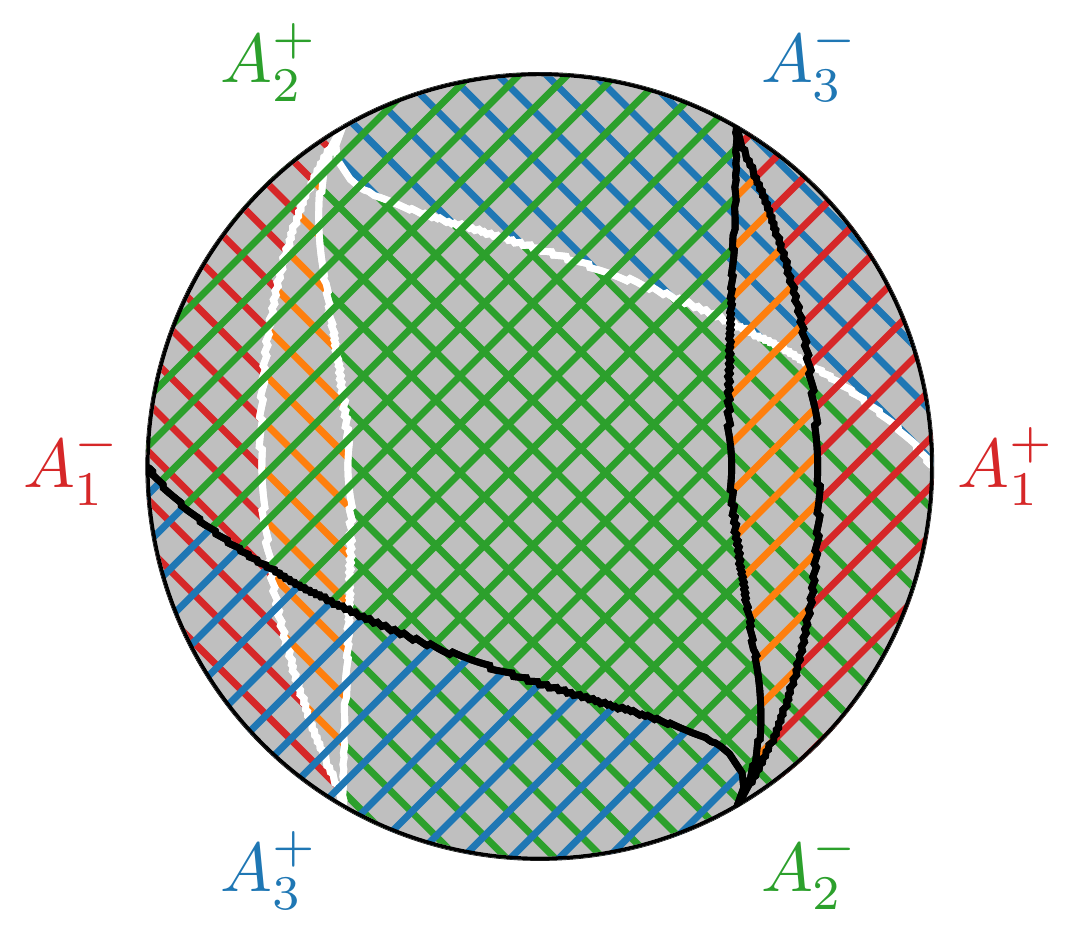}
    \caption{The equatorial plane with parameter values $\alpha_1 = 0.5, \alpha_2=0.7, \alpha_3 =-0.2$ ($++-$ case). Coloring as in Figure 
    \ref{fig:ppp}. } 
     \label{fig:ppm}
\end{figure}

\begin{figure}[!h]
    \centering
    \includegraphics[width=0.6\textwidth]{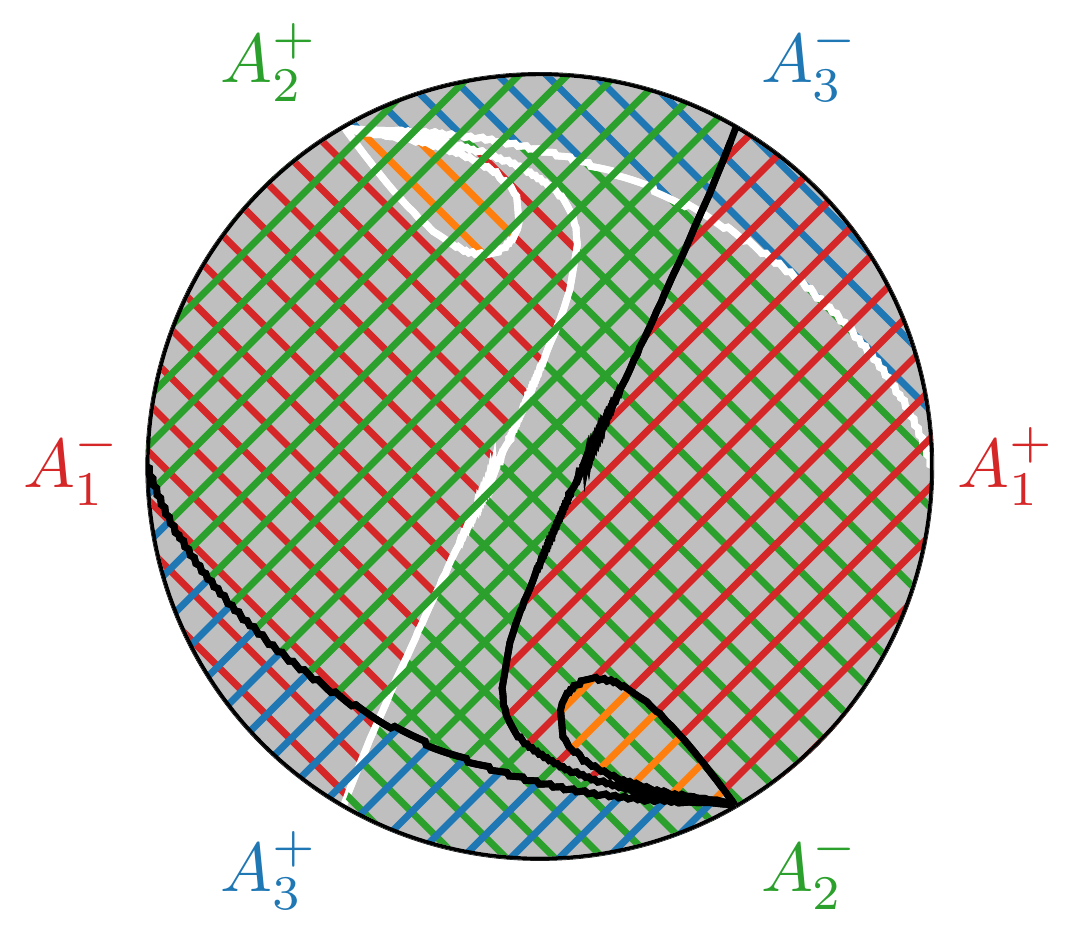}
    \caption{The equatorial plane with parameter values $\alpha_1 = 1.1, \alpha_2=-0.03, \alpha_3 =-0.07$ ($+--$ case).  Coloring as in Figure 
    \ref{fig:ppp}. }
      \label{fig:pmm}
\end{figure}

\def\r{\textwidth/2.01} 

The $B$ to $B$ solutions are of some significance. At any of the $B$ points, one of the components $f_1,f_2,f_3$ diverges, but the others tend to zero.
Thus one component of a $B_i^-$ to $B_j^+$ solution, with $i\not= j$, 
gives a solution of $P_\mathrm{IV}$ that is not only nonsingular on the entire real axis, but also tends to
$0$ as $x\rightarrow \pm \infty$.  Our results give methods for searching for these solutions, as they sit on the intersection of one $B^+$ stable manifold
and one $B^-$ unstable manifold, thus corresponding to an initial value at which both the $x\rightarrow +\infty$ and the $x\rightarrow -\infty$ asymptotics 
changes. We show some relevant numerical results in Appendix B. 
Figures \ref{fig:pppzoom},\ref{fig:+++BtoB} illustrate in the $+++$ case. Figures  \ref{fig:ppmzoom},\ref{fig:++-BB1},\ref{fig:ppmzoom2},\ref{fig:++-BB2} illustrate for the two different types of intersection point that occur in the $++-$ case.
And Figures \ref{fig:pmmzoom},\ref{fig:+--BB1} illustrate in the $+--$ case. 
Note that in certain cases the $B$ to $B$ solutions have the further property of having no zeros on the entire real axis. 

Finally in this section, we show some numerical results on the shape of the $C^-$ to $C^+$ region that exists in the $+++$ case. The numerical 
evidence we have points to there being only a single open region of such solutions in the space of initial data. In Figure   \ref{fig:sel} we plot the boundary
of the relevant region in the space of initial data, for a variety of parameter values. As expected, we observe that as any of the parameters $\alpha_i$ gets smaller, the area of the region contracts. 


\begin{figure}[!h]
    \centering
    \includegraphics[width=\textwidth]{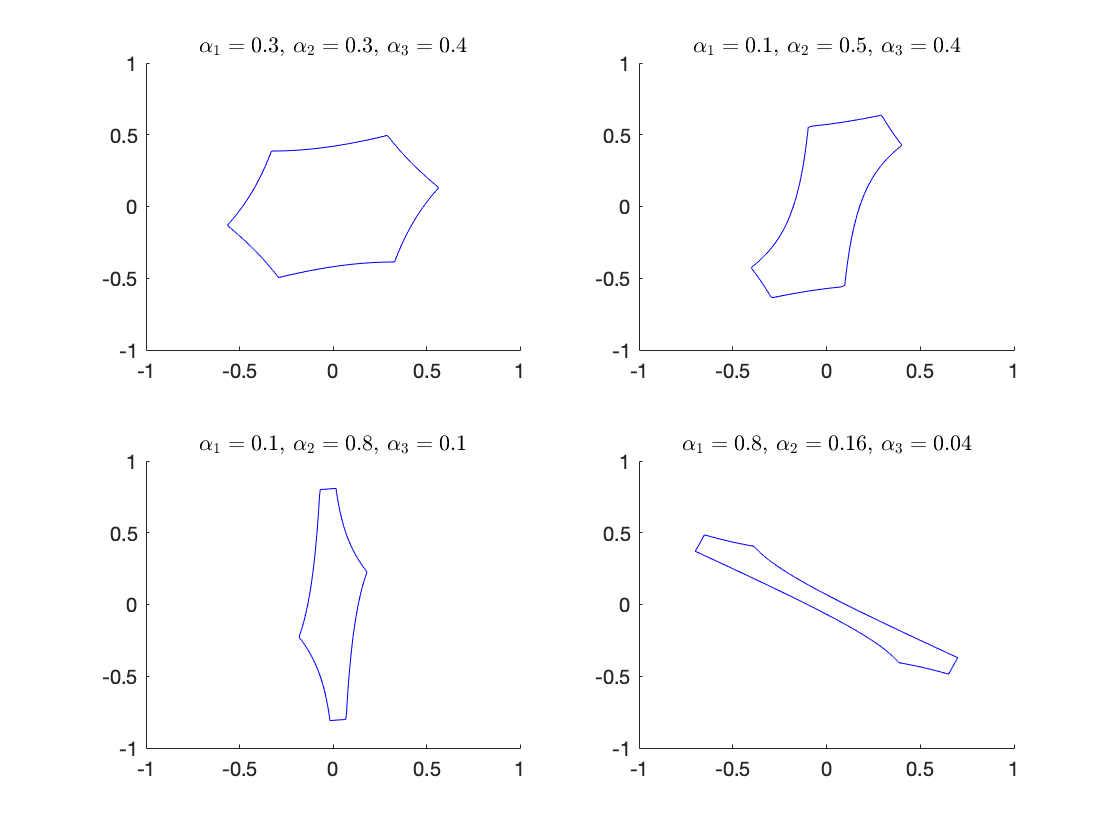}
    \caption{The region of $C$ to $C$ solutions in the space
    of initial values, for different sets of parameters in the $+++$ case. 
    The horizontal axis is $f_1(0)$, the vertical axis is $\frac{f_2(0)-f_3(0)}{\sqrt{3}}$ 
    and $f_1(0)+f_2(0)+f_3(0)=0$.}  \label{fig:sel}
\end{figure}

\section{Solutions with Poles and Allowed Pole Sequences}

In this section we consider solutions of $sP_\mathrm{IV}$ that have singularities. There clearly are four possibilities: a solution could have a
finite sequence of singularities, or a singly infinite sequence with no singularities for $x$ less than a certain finite value,
or a singly infinite sequence with no singularities for $x$ more  than a certain finite value,
or a doubly infinite sequence. In the first  case, we need to specify the asymptotics of the solution as $x\rightarrow\pm\infty$,
in the second case as $x\rightarrow -\infty$, and in the third case as $x\rightarrow +\infty$. We restrict ourselves to the generic case
of type $C$ asymptotic behavior. The discussion can be extended to include type $B$ asymptotic behavior as well,  but we do not pursue this here. 

In the obvious manner we associate with each solution of this type a symbol sequence specifying the singularities and the asymptotic behavior.
In the case of a finite sequence of singularities, the sequence begins and ends with $C$ and has a finite sequence of the symbols $A_1,A_2,A_3$
between the two $C$s. In the singly infinite case the sequence will begin or end with $C$, followed or preceded by an infinite sequence of $A$s.
In the doubly infinite case the sequence just consists of $A$s.  Table \ref{tab:excls} shows that certain symbols cannot follow certain other symbols, 
depending on the signs of the parameters $\alpha_i$.
We recall that these forbidden transitions are obtained by considering the signs of $f_1,f_2,f_3$ near the various singularities and in the appropriate asymptotic regimes
(see  equations (15) and (16) in \cite{SchiffTwiton}),  and using the fact that the sign of $\alpha_i$ determines the sign of $f_i'$ at a
regular zero of $f_i$ (that is, at a zero where all components of the solution are nonsingular); thus between any two singularities, there can be at most one
regular zero of each of the $f_i$, and the change in sign of $f_i$ at such a zero can only be in a specific direction.
In the cases of permitted transitions, Table \ref{tab:excls} also  shows which of the functions $f_1,f_2,f_3$ have zeros (though note that the order of these zeros is
not determined). 

We now prove the following result:
\begin{theorem*}
In the case $\alpha_1,\alpha_2,\alpha_3>0$,
\begin{enumerate}
\item  The only permitted finite singularity sequence is $CC$. 
\item  The only permitted singly infinite singularity sequences are 
\begin{equation}
  \begin{array}{r}
C  A_2 A_1 A_3 A_2 A_1 A_3 A_2A_1   \ldots  \\ 
  C     A_1 A_3 A_2 A_1 A_3 A_2A_1   \ldots  \\
    C        A_3 A_2 A_1 A_3 A_2A_1   \ldots   
   \end{array}  \label{eq:si1}
\end{equation}
and
\begin{equation} \begin{array}{l}
     \ldots A_1A_2 A_3 A_1 A_2 A_3A_1 A_2 C    \\ 
     \ldots A_1A_2 A_3 A_1 A_2 A_3A_1 C   \\ 
     \ldots A_1A_2 A_3 A_1 A_2 A_3 C     
     \end{array}  \label{eq:si2}  \end{equation}
\item  The only permitted doubly infinite singularity sequences are 
 \begin{equation} \begin{array}{c}
        \ldots.  A_1A_2A_3 A_1A_2A_3 ~~~~~~~  A_1  ~~~~~~~~A_3A_2A_1A_3A_2A_1\ldots  \\
        \ldots.  A_1A_2A_3 A_1A_2A_3 ~~~~ A_1A_2A_1 ~~~~  A_3A_2A_1A_3A_2A_1\ldots     \\
        \ldots.  A_1A_2A_3 A_1A_2A_3  A_1A_2A_3A_2A_1     A_3A_2A_1A_3A_2A_1\ldots     \\ 
\end{array}    \label{eq:si3}\end{equation}
  or doubly infinite repetitions of the subsequences $A_1A_2A_3$ or $A_3A_2A_1$.  
\end{enumerate}

\end{theorem*}

\begin{proof}  From Table \ref{tab:excls}, in the case that all the $\alpha_i$ are positive, 
a singularity of type $A_i$ cannot be followed by another singularity of type $A_i$.

To obtain further restrictions on the permitted sequence of singularities, we consider the action
of the symmetry group. Note that both the symmetries $\sigma$ and $\tau$ described in Section 2.2 preserve
asymptotic type $C$ behavior. The action of $\sigma$ maps singularities of type $A_i$ to singularities of type
$A_{i+1~{\rm mod}~3}$ and regular zeros of $f_i$ to regular zeros of $f_{i+1~{\rm mod}~3}$. A brief calculation shows that the action
of $\tau$ does not affect type $A_2$ and type $A_3$ singularities, but eliminates type $A_1$ singularities, leaving
regular  zeroes of $f_1$ at the points of singularity; at the same time it  creates new type $A_1$ singularities out of
regular zeros of $f_1$ (and this is the only way the action of $\tau$ can create new singularities). 
Recall also that $\tau(\alpha_1,\alpha_2,\alpha_3) = (-\alpha_1, \alpha_2+\alpha_1,\alpha_3+\alpha_1)$
so after the action of $\tau$ the new value of $\alpha_1$ is negative, and the new values of $\alpha_2,\alpha_3$ are still
positive.

Suppose that a solution with $\alpha_1,\alpha_2,\alpha_3>0$ has symbol sequence $CA_1C$. From Table \ref{tab:excls}
 we see that there
are zeros of $f_1$ both between the first $C$ and the $A_1$, and between the $A_1$ and the second $C$. Thus applying $\tau$ 
gives the symbol sequence $CA_1A_1C$. But now $\alpha_1<0$ and $\alpha_2,\alpha_3>0$, and we see from Table \ref{tab:excls}
 that in this
situation an $A_1$ singularity cannot follow another $A_1$ singularity. Thus we have a contradiction, and the symbol sequence
$CA_1C$ is not permitted when $\alpha_1,\alpha_2,\alpha_3>0$. Similar arguments eliminate any symbol sequence containing
any of the subsequences $CA_1A_2$, $A_2A_1C$ or $A_2A_1A_2$.  Application of $\tau$ to any of these will lead to two
consecutive $A_1$ singularities, which is not allowed. 

Similarly, applying $\tau_2 = \sigma^2 \tau \sigma$ (which maps the parameters to $\alpha_1+\alpha_2, -\alpha_2, \alpha_3+\alpha_2$, i.e.
to the $+-+$ case) eliminates the substrings $CA_2C$, $CA_2A_3$, $A_3A_2C$, $A_3A_2A_3$. 
and applying $\tau_3=\sigma \tau\sigma^2$ eliminates the substrings $CA_3C$, $CA_3A_1$, $A_1A_3C$, $A_1A_3A_1$. 
 
Let us now consider  a permitted sequence starting $CA_1$. The sequences  $CA_1C$, $CA_1A_2$ and $A_1A_1$  are not allowed, so the sequence must
in fact start with $CA_1A_3$. The sequences $A_1A_3A_1$, $A_1A_3C$ and $A_3A_3$ are not allowed, so the sequence  must in fact start $CA_1A_3A_2$.
Continuing, we reach the conclusion that the only permitted sequence starting $CA_1$ is the singly infinite sequence
$CA_1A_3A_2A_1A_3A_2A_1\ldots$.
Similarly we deduce that the only permitted sequence starting $CA_2$ is  $CA_2A_1A_3A_2A_1A_3A_2A_1\ldots$
and the only permitted sequence starting $CA_3$ is $CA_3A_2A_1A_3A_2A_1A_3A_2A_1\ldots$.

This proves the section of the theorem relating to sequences (finite or singly infinite) starting with $C$.
The section relating to singly infinite sequences ending in $C$ is proved similarly. For doubly infinite sequences, it is
straightforward to show that the subsequences $A_3A_1A_3$, $A_1A_2A_1$, $A_2A_3A_2$ have unique doubly infinite extensions, and
the only other possible doubly infinite sequences of $A$s that are not excluded are doubly infinite repetitions of one of the subsequences $A_1A_2A_3$
or $A_3A_2A_1$.  
\end{proof}

\begin{figure}[t]
	\centering
    \includegraphics[width=0.7\textwidth]{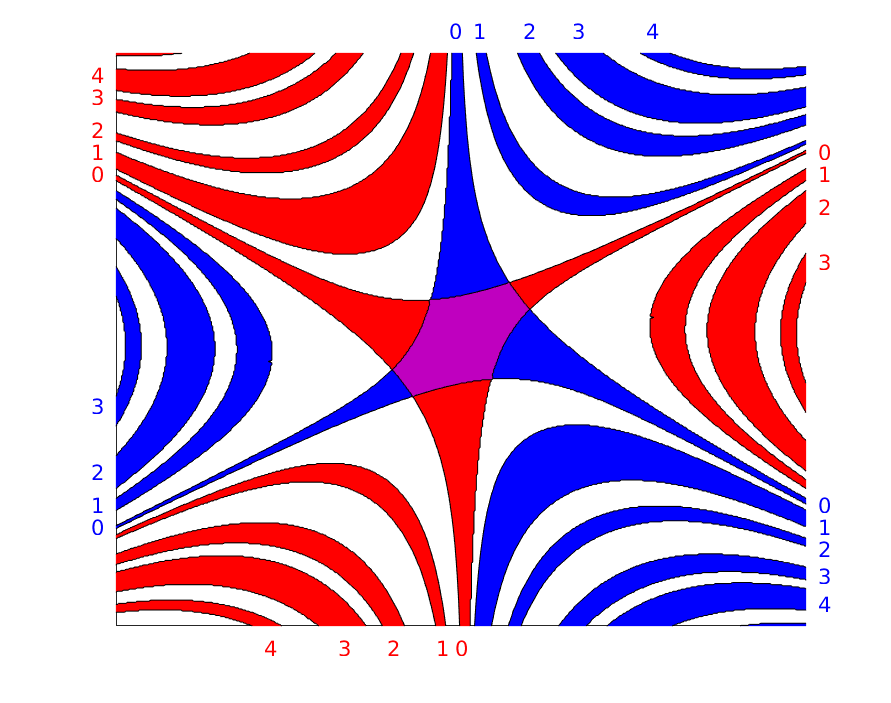}
	\caption{Types of solution for different initial conditions. Parameter values are  $\alpha_1=0.2,\alpha_2=0.3,\alpha_3=0.5$. The horizontal axis is $f_1(0)$, the vertical axis is $\frac{f_2(0)-f_3(0)}{\sqrt{3}}$, both ranging approximately from $-3$ to $3$. Purple region: nonsingular solutions. Blue regions: An infinite  number of singularities in $x<0$ and a finite number in $x>0$, with the number of singularities in $x>0$ marked for each region. Red regions: the same, reversed. White regions: an infinite sequence of singularities in both $x>0$ and $x<0$.}
	\label{fig:regions}
\end{figure}

 Numerical experiments show that in practice there exist solutions with no singularities, as we have already documented; there also exist solutions with 
 all the possible singly infinite singularity sequences, and solutions with the first 3 types of doubly infinite singularity sequence. We have not yet 
 found evidence of the last two possibilities, viz. doubly infinite repetitions of one
 the subsequences  $A_1A_2A_3$ or $A_3A_2A_1$. In  Figure 
 	\ref{fig:regions} we show numerical results for one specific choice of 
 	the parameters (using the numerical method explained in Appendix A for integrating through poles). For different choices of initial condition 
 	at $x=0$ we integrate up to $x=10$ and down to $x=-10$ and count the number of poles in $x>0$ and in $x<0$. For the purpose of the experiment any number of poles exceeding $10$ is considered to be infinite. We find ``bands'' in the
 	plane of initial values with $0,1,2,3,\ldots$ poles in $x>0$, and 
 	corresponding bands (related by the symmetry $f(x)\rightarrow -f(-x)$) with 
 	$0,1,2,3,\ldots$ poles in $x<0$. The only intersection of the bands 
 	is apparently a single region giving pole-free solutions. Between the 
 	bands we observe regions where there are presumably solutions with 
 	doubly infinite singularity sequences. In the regions corresponding to 
 	singly infinite singularity sequences we observe cases with the ``last'' 
 	singularity being of all $3$ possible types, as shown in 
 	(\ref{eq:si1})  and 	(\ref{eq:si2}); in the doubly infinite 
 	bands we observe cases with the ``central'' singularity being of all 
 	$3$ possible types given in  	(\ref{eq:si3}). 
 	However, we do not observe cases where the singularity sequence is a 
 	doubly infinite repetition of the subsequences $A_1A_2A_3$
 	or $A_3A_2A_1$. 

 	Many related plots 
 	can be found in the numerical work of Reeger and Fornberg
 \cite{reeger2013painleve,  reeger2014painleve}.  Much of Reeger
 and Fornberg's work relates to nongeneric parameter values, but  Figure 8 in  \cite{reeger2014painleve} relates to the case $\alpha_1=\alpha_2=\frac14$, $\alpha_3=\frac12$, and shows what 
 appears to be two regions in the space
 of initial values for which there is a pole-free solution. This
 arises as $u(0)$ and $u'(0)$ in Reeger and Fornberg's work correspond to 
 $f_1(0)$ and $-\frac12 + f_1(0)(f_2(0)-f_3(0))$ in our work. 
 	
The results for the case that $\alpha_1,\alpha_2,\alpha_3$ 	 are all 
positive can be generalized using the transformation group to appropriate results
for any generic choice of parameters. To see the effect of the transformation
$\tau$ on a particular singularity sequence (for a particular set of parameters), we use Table \ref{tab:excls} to locate the zeros of $f_1$ (inserting an appropriate symbol, say $Z_1$), then we delete the existing $A_1$s and replace the $Z_1$s by new $A_1$s. The effect of the transformation $\sigma$ is simply to cycle $A_1,A_2,A_3$. Clearly doubly infinite sequences remain doubly infinite, singly infinite sequences remain singly infinite, and finite sequences remain finite. Thus we arrive at the results stated in the Introduction:  {\em for any generic choice of the parameters, there exists a unique finite
sequence of singularities for which $sP_\mathrm{IV}$ has a two-parameter 
family of solutions with this singularity sequence}. The (singly and doubly) 
infinite sequences that are permitted in general also depend on the values of the parameters. For example, by application of $\tau$ to the sequence 
$$ \ldots A_1A_2A_3 A_1A_2A_3 ~~~~ A_1 ~~~~ A_3A_2A_1 A_3A_2A_1  \ldots $$ 
that is permitted in the $+++$ case, we obtain the same sequence with the ``central'' $A_1$ removed, i.e. 
$$ \ldots A_1A_2A_3 A_1A_2A_3 ~~~~~~~~ A_3A_2A_1 A_3A_2A_1  \ldots $$ 
(note that a repeated $A_3$ is allowed in the case $-++$). 
The two sequences consisting of doubly infinite repetitions of $A_1,A_2,A_3$ or 
$A_3,A_2,A_1$  seem to be allowed for arbitrary values of the parameters. 
They are evidently invariant under the action of $\sigma$, and in fact 
are also invariant under the action of $\tau$ except in the case $--+$. 
In this case both are transformed to the sequence which is a doubly infinite
repetition of the subsequence $A_2A_3$ (which seems to be allowed in the cases 
$+-+$ and $+--$, but we have not observed its existence).  

\section{Special Solutions} 

\subsection{Rational Solutions 1}

In this section we consider the rational solutions of   $sP_\mathrm{IV}$
that are equivalent to the so-called ``$-\frac23 z$ hierarchy'' of 
$P_\mathrm{IV}$ \cite{clarkson2003fourth}.   These occur when the parameters
take values at the centers of faces of the lattice in Figure \ref{fig:lattice}.
The fundamental example is the solution $f_1=f_2=f_3=\frac13 x$ that is obtained 
when $\alpha_1=\alpha_2=\alpha_3 =\frac13$.   Indeed the general solution of this kind is obtained by application of an arbitrary element of the symmetry group 
on this fundamental solution \cite{murata,noumi1999symmetries}. Thus for example, by applying $\tau\sigma^2\tau$ to the fundamental solution we obtain the solution 
$$
f_1 = \frac{x^2-3}{3x} \ , f_2 = \frac{x(x^2+3)}{3(x^2-3)} \ , 
   f_3 = \frac{x^4-6x^2-9}{3x(x^2-3)} \ , 
$$
for parameter values $\alpha_1=-\frac23,\alpha_2=\frac13,\alpha_3=\frac43$. 
In the context of this paper, this result is useful as by looking at the 
singularities of the rational solution we can determine the 
unique finite singularity sequence of solutions for {\em all} values of the 
parameters in the cell whose center gives the rational solution. In the case
of the example just given, there are singularities at $x=0,\pm \sqrt{3}$ and the 
sequence is $CA_1A_2A_1C$. This sequence can be read off from the group element 
that generated the solution. 
The fundamental solution has singularity 
sequence $CC$, applying $\tau$ gives the sequence $CA_1C$ (and parameter values $-++$), applying 
$\sigma^2$ gives the sequence $CA_2C$ (and parameter values $+-+$), and 
applying $\tau$ again gives the sequence $CA_1A_2A_1C$.  Note that in the context of
studies of the distribution of singularities of rational solutions  in the complex plane \cite{cfilipuk,cmilson}, no rules are currently known for evaluating the effect of the transformation $\tau$. 

We also note the following fact: If the action of the group element $g$ on the 
fundamental solution gives $f$ for parameter values $\alpha$, then using the 
obvious notation we have 
$$ ( g^{-1} f )_1 = ( g^{-1} f )_2 = ( g^{-1} f )_3 \ .  $$
This gives two polynomial equations that must be satisfied by the components of
the rational solution $f$, in addition to the constraint $f_1+f_2+f_3=x$. Thus, in
the above example, it is possible to check that 
\begin{eqnarray*}
9 f_1^2 f_2^2-9f_1^2 f_2 f_3+3f_1^2-18f_1f_2+6f_1f_3+8
    &=& 0 \\
-9f_1^3f_2 + 9f_1^2f_2f_3 + 6f_1f_2 - 6f_1f_3 - 4    &=& 0  
\end{eqnarray*}
As far as we are aware, these polynomial relations between the components
of a rational solution are new.

\subsection{Rational Solutions 2}
In this section we consider the rational solutions of   $sP_\mathrm{IV}$
that are equivalent to the so-called ``$-2 z$ hierarchy'' 
and ``$-\frac1{z} $ hierarchy''   of 
$P_\mathrm{IV}$ \cite{clarkson2003fourth}.   These occur when the parameters
take values at the vertices of the lattice in Figure \ref{fig:lattice}.
The simplest example is the solution $f_1=x$, $f_2=f_3=0$ that is obtained 
when $\alpha_1=1$, $\alpha_2=\alpha_3 =0$.  A complication arises that 
does not occur for the first set of rational solutions: Applying the transformation $\tau$ requires that $f_1$ should be nonzero, and $f_1$ can be zero
in the case that $\alpha_1=0$. The general rational solution of the 
second type arises by application of a restricted set of group elements to 
the fundamental solution, those that avoid generating the parameter value
$\alpha_1=0$ as an intermediate step. However, there is no shortage 
of group elements with this property. The resulting solutions also 
satisfy polynomial identities, in this case obtained from the conditions 
$$ ( g^{-1}f )_2 =  ( g^{-1}f )_3 = 0 \ .  $$
As an example, application of the group element  
$\sigma\tau\sigma^2\tau \sigma^2\tau\sigma\tau\sigma\tau $  to the fundamental
solution gives 
$$
f_1 = \frac{2x(x^2-3)(x^2+1) }{(x^2-1)(x^4+3)}\  ,  
f_2 = \frac{-2x(x^2-1)(x^2+3)}{(x^2+1)(x^4+3)}\ ,   
f_3 = \frac{x(x^4+3)}{(x^2-1)(x^2+1)}\ ,
$$
for parameter values $\alpha_1=\alpha_2=2,\alpha_3=-3$. These functions obey
the polynomial identities 
\begin{eqnarray*}
  f_1^2f_2f_3^2 + f_1^2f_3 - 5f_1f_2f_3 - f_1f_3^2 - 3f_1 + 6f_2 + 2f_3  &=& 0\ , \\
  f_1f_2^2f_3^2 + 5f_1f_2f_3 - f_2^2f_3 + f_2f_3^2 + 6f_1 - 3f_2 + 2f_3  &=& 0\ .
\end{eqnarray*}
Although we have not discussed solutions with singularities with $B$ type asymptotics in this paper, we mention that this solution has singularity 
sequence $B_3A_2A_2B_3$. 

It is straightforward to establish  that there is a rational solution  of  $sP_\mathrm{IV}$
with $f_1=0$ (and thus $\alpha_1=0$) for arbitrary integer values of 
$\alpha_2,\alpha_3$ with $\alpha_2+\alpha_3=1$.  For positive integer values of 
$\alpha_2$ the solution has
$$ f_2(x) = 
\frac
{\sqrt{2} \alpha_2 {\rm He}_{\alpha_2} \left( \frac{ix}{\sqrt{2}} \right) }
{i  {\rm He}_{\alpha_2}' \left( \frac{ix}{\sqrt{2}} \right)}\ ,  f_3 = x-f_2\ ,
$$
where ${\rm He}_n$ denotes the $n$th Hermite polynomial. This has singularity 
sequence $B_2B_2$ for odd values of $\alpha_2$  and $B_2A_1B_2$ for even values
of $\alpha_2$.  
For nonpositive integer values of 
$\alpha_2$ the solution is 
$$ f_2(x) = 
-\frac
{  {\rm He}_{-\alpha_2}' \left( \frac{x}{\sqrt{2}} \right)}
{\sqrt{2} {\rm He}_{-\alpha_2} \left( \frac{x}{\sqrt{2}} \right) }
\ ,  f_3 = x-f_2 \ .
$$
This has singularity sequence  $B_3A_1\ldots A_1 B_3$ with $-\alpha_2$ 
successive 
singularities of type $A_1$. Note that since these solutions have $\alpha_1=0$
the rules of Table \ref{tab:excls} do not apply. However, since  
successive $A_1$ singularities are allowed in both the cases $+-+$ and $--+$, it
is not surprising that we also see this on the transition between them.  Note 
also that these solutions have all their singularities on the real axis, with no further poles in the complex plane. 

\subsection{Other Special Solutions}

In greater generality, whenever $\alpha_1=0$,  $sP_\mathrm{IV}$ has a 
1-parameter family of solutions with $f_1=0$. In this case the  $sP_\mathrm{IV}$
system reduces to the first order Riccati equation
$$  f_2' = f_2(x-f_2)+ \alpha_2\ , $$
which can be linearized. By application of suitable group elements, these 
solutions give rise to a 1-parameter family of solutions whenever any of 
the $\alpha_i$ takes an integer value.  In fact, these solutions (and the 
corresponding solutions of   $P_\mathrm{IV}$) can all be obtained as the 
solution of a first order differential equation. To see this, note that 
by application of the inverse group element to the solution, it must 
satisfy the single polynomial identity 
$$   ( g^{-1}f )_1 = 0 \ .  $$
Substituting 
$$
f_2 = \frac12 \left( x-f_1 + \frac{f_1'-\alpha_1}{f_1} \right)  \ , \qquad 
f_3 = \frac12 \left( x-f_1 - \frac{f_1'-\alpha_1}{f_1} \right)
$$
gives a first order differential equation for $f_1$.
Thus, for example, by applying the  transformation 
$\sigma^2\tau\sigma^2\tau\sigma\tau   \sigma^2\tau\sigma^2\tau\sigma   $
to the solutions with $f_1=0$ we obtain the set of 
special solutions with $\alpha_1=2$,
and these give rise to the special solutions of $P_\mathrm{IV}$, equation 
(\ref{eq:p4}), with $\beta=-8$, that satisfy the first order differential 
equation 
\begin{eqnarray*}
 0&=&   \left(\frac{\mathrm{d}w}{\mathrm{d}z}\right)^{4}
 +   8\,\left(\frac{\mathrm{d}w}{\mathrm{d}z}\right)^{3}
 +   \left( - 2w^4-8z w^3-8(z^2-\alpha) w^2\right) 
      \left(\frac{\mathrm{d}w}{\mathrm{d}z}\right)^{2} 
 +   \left( -8w^{4}-32z w^3
     -32(z^2-\alpha)w^2 
-128 \right) \left(\frac{\mathrm{d}w}{\mathrm{d}z}\right) \\ 
&&
+ w^{8}+8z w^7 + 8(3z^2-\alpha) w^6 + 32z(z^2-\alpha) w^5 + 16\left( 
(z^2-\alpha)^2 +1 \right)   w^4 
+64 z w^3   -256 
\end{eqnarray*}
Thus there are first order equations of higher and higher degree that 
are consistent with $P_\mathrm{IV}$ for suitable values of the parameters. 
This generalizes the old observation that for suitable values of the 
parameters $P_\mathrm{IV}$ is consistent with a Riccati equation 
 \cite{FA}. 
 
 \section{Concluding Remarks}
 
 In the course of this paper a framework has emerged for classification of
 real solutions of $sP_\mathrm{IV}$ (and thus also for $P_\mathrm{IV}$):  a solution  is classified by its asymptotic behavior as $x\rightarrow\pm\infty$ and its  singularity sequence, with the asymptotic behavior being 
 superfluous in one or both limits in the cases of singly infinite or doubly 
 infinite singularity sequences respectively.  We have established a strong 
 result for the existence of solutions with no singularities. For the case 
 of nonzero parameter values, solutions exist exhibiting all the transitions 
 allowed by Table \ref{tab:excls2}.  For the case of solutions with singularities
 and with generic ($C$-type) asymptotic behavior (if needed), we have given a
 list of the possible singularity sequences in the $+++$ parameter case, from which a similar list can be derived for an arbitrary generic  (noninteger) set of parameters. In particular we have seen that for any generic set of 
 parameters there is a unique allowed finite singularity sequence for solutions with $C$ to $C$ asymptotics. Numerics in the $+++$ case indicate that all the permitted singularity sequences actually occur, with the possible exception 
of doubly infinite repetitions of the subsequences $A_1A_2A_3$ and $A_3A_2A_1$.  
 We are hopeful that it might be possible to exclude these possibilities using techniques not considered in the current  paper; it is well-known that there 
 are solutions of $P_\mathrm{IV}$ with elliptic function asymptotics for large argument  \cite{vereshchagin}, and this is a question about the connection formulae for these solutions. 
 
 Our work has all been on the basis of an assumption concerning the dynamical system described in our previous work  \cite{SchiffTwiton}. We are happy with this assumption as there is no evidence to the contrary, and it is an assumption of the simplest possible scenario (that the only possible asymptotics are the $B$ and $C$ behaviors we have described). However proving it looks difficult, 
 as  it involves the local stability properties of a periodic orbit in the case that the linearized approximation gives insufficient information. 
 
 We believe the approach given here for $P_\mathrm{IV}$ should be extendable to 
 other Painlev\'e equations. Relevant dynamical systems have been given in \cite{adler,wh}. However, the works of Chiba \cite{chiba,chiba2} suggest that more subtle compactifications will be involved. 
 
 
\appendix

\section{Numerical Methods for Integrating Through Poles}

We use the following simple idea to construct changes of the dependent 
variables that allow us to integrate through the three types of pole 
singularity of   $sP_\mathrm{IV}$. Near the $A_3$ type singularity 
the system has a Painlev\'e series as given by Equation (15) in 
\cite{SchiffTwiton}. This expresses $f_1,f_2,f_3$ in terms   of three
quantities $x-x_0, x_0, C$ all of which remain finite near the pole. 
Truncating this expansion in such a way that $f_1$ depends only on the first 
quantity (which we call $z_1$), $f_2$ depends only on the first and second
(which we call $2z_2$), and $f_3$ depends on the first, the second and the third
(which we call $-2z_3$) but on the latter only linearly, gives the substitution
 \begin{eqnarray*}
    f_1 &=& \frac1{z_1}   \ , \\
    f_2 &=&  - \frac1{z_1} + z_2    \ , \\
    f_3 &=& -\alpha_3  z_1  +  z_3 z_1^2 \ .     
\end{eqnarray*}
This is, by construction, an invertible change of variables, with inverse
\begin{eqnarray*}
    z_1 &=& \frac1{f_1}\ ,    \\
    z_2 &=&  f_1 + f_2  \ ,   \\
    z_3 &=& \alpha_3 f_1 +  f_1^2f_3\ .       
\end{eqnarray*}
The variables $z_1,z_2,z_3$ satisfy the system 
\begin{eqnarray*}
    z_1' &=&  1 + z_1 \left(z_1^2 z_3 - z_1(\alpha_1+\alpha_3) -z_2 \right)\ , \\
    z_2' &=& 1+ \alpha_3 + z_1 \left( z_1z_2z_3 - (2+\alpha_3)z_3   \right) \ ,   \\
    z_3' &=& z_2z_3  - \alpha_3(\alpha_1+\alpha_3) + z_1z_3 \left( 
              2\alpha_1 + 3\alpha_3 - 2 z_1z_3  \right) \ . 
\end{eqnarray*}
As soon as a pole of the appropriate type in the $f$ system is approached (say if  $|f_1|,|f_2|>10$) we change to the $z$ variables and integrate there until the pole  is passed.  Similar changes of variables are used near the other 
two types of pole. 

\section{Numerics for $B$ to $B$ solutions}

As visible from Figures \ref{fig:ppp},\ref{fig:ppm},\ref{fig:pmm}, 
solutions with $B$ to $B$ type asymptotics occur at the crossing points 
of the stable manifolds of the $B^+$ points with the unstable manifolds
of the $B^-$ points. It is possible to numerically search for these solutions 
by locating four points in the four regions adjacent to the crossing, 
distinguished by the associated asymptotic behaviors as $t\rightarrow\pm\infty$, 
and then recursively reducing the size of the associated quadrilateral (as 
measured by its perimeter) until the crossing point is found to sufficient 
accuracy.  The following plots show some examples. 
Figures \ref{fig:pppzoom},\ref{fig:+++BtoB} are relevant to one of the 
$6$ crossing points in the $+++$ case. Figure \ref{fig:pppzoom} shows the different  asymptotics in the four adjacent regions,  and \ref{fig:+++BtoB}
shows the $B$ to $B$ solution once the initial condition has been found to 
sufficient accuracy to give an accurate plot on the interval $[-10,10]$. 
Figures  \ref{fig:ppmzoom},\ref{fig:++-BB1},\ref{fig:ppmzoom2},\ref{fig:++-BB2} illustrate for the two different types of intersection point that occur in the $++-$ case. In one case the resulting $B$ to $B$ solution has no zeros in any components, in the other there is a zero in one component (see Table 
\ref{tab:excls2}). Finally, 
Figures \ref{fig:pmmzoom},\ref{fig:+--BB1} illustrate in the $+--$ case. 

\newpage

\tikzset{mark size=0}

\begin{figure}[!h]
	\centering
    \includegraphics[width=0.55\textwidth]{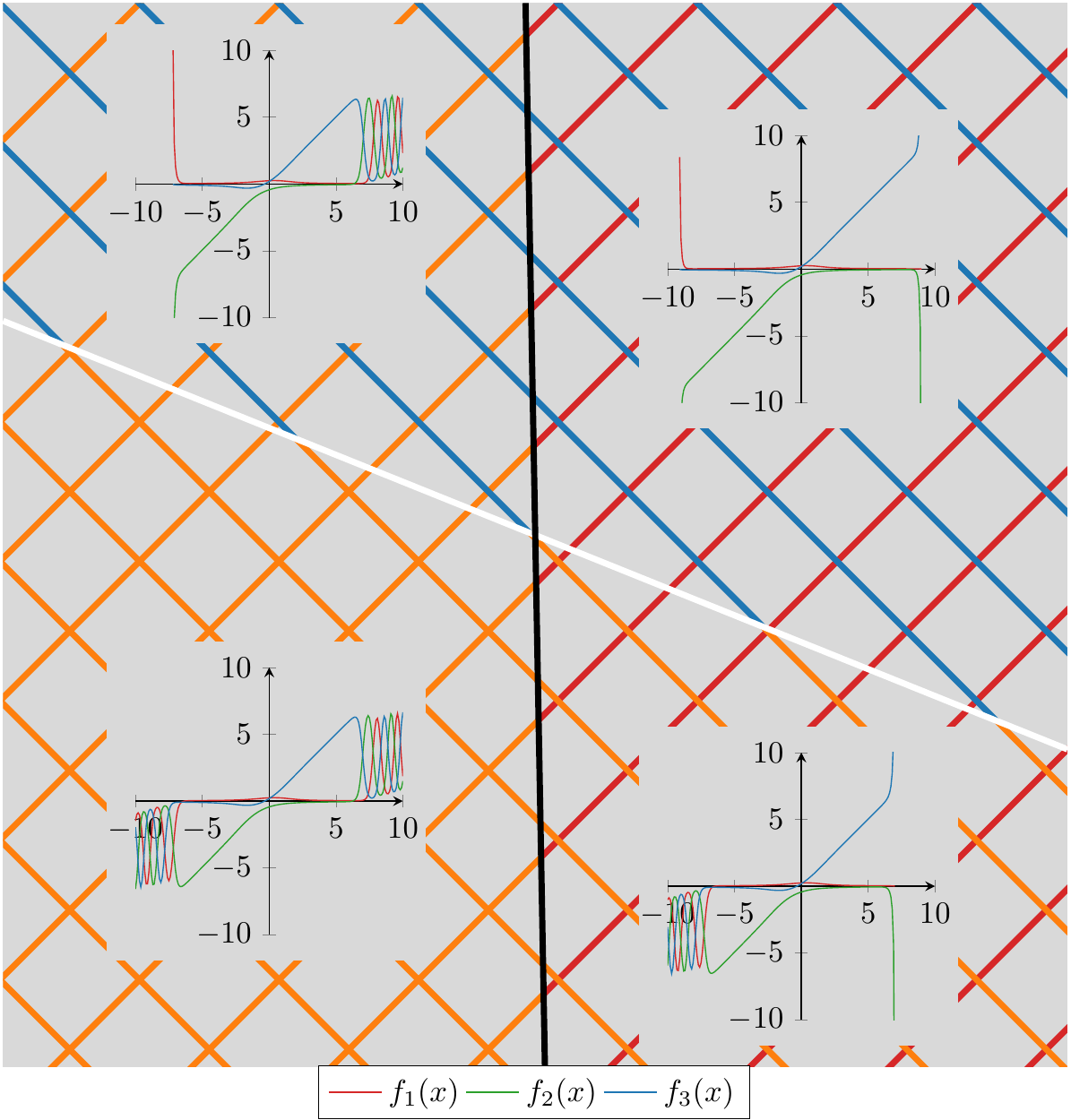}
	\caption{A zoomed-in version of the first-quadrant intersection in Figure \ref{fig:ppp}.}
	\label{fig:pppzoom}
\end{figure}	

\begin{figure}[!h]
    \centering
    \includegraphics[width=0.55\textwidth]{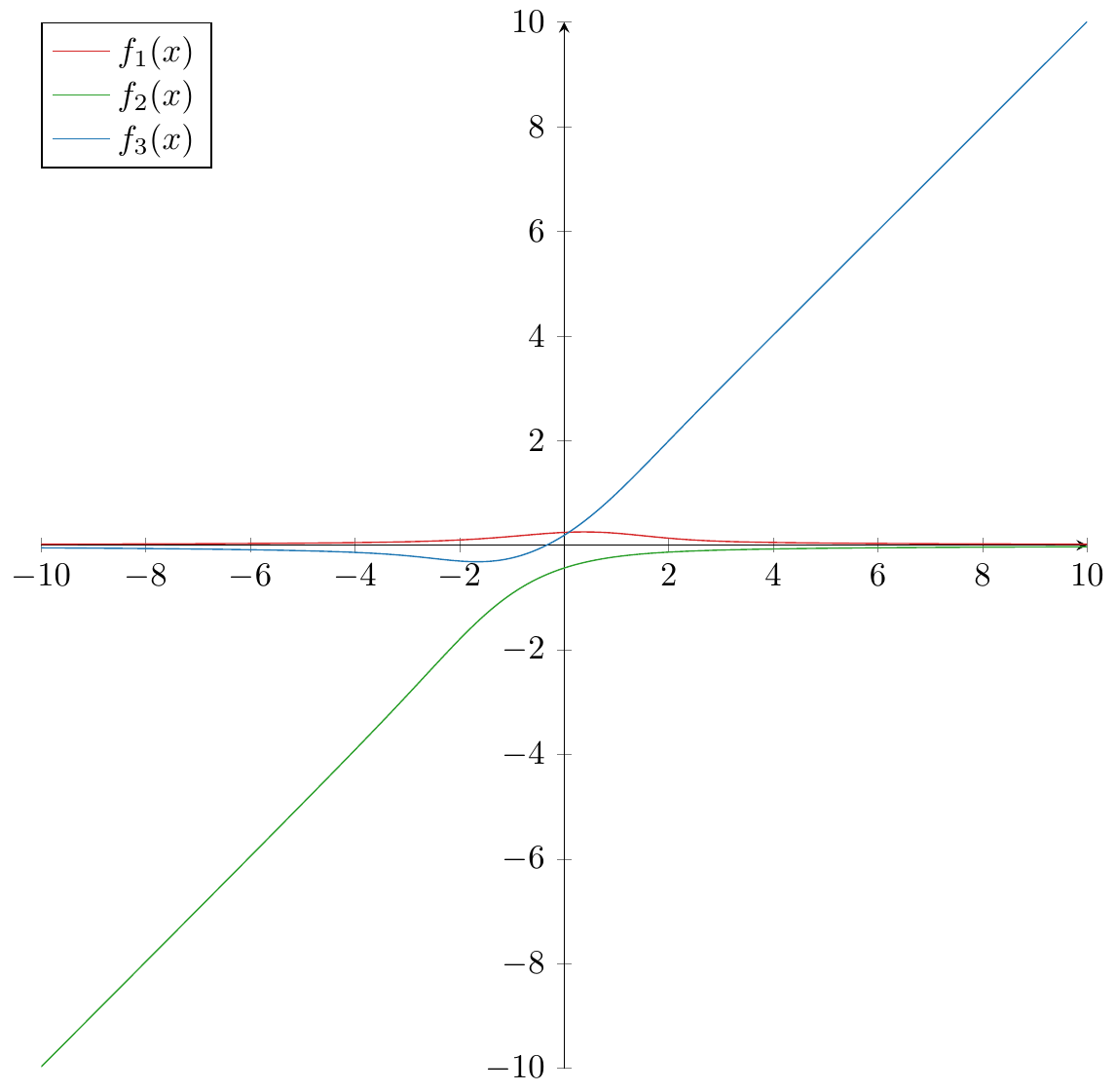}
    \caption{The $B_2^{-} \to B_3^{+}$ solution corresponding to the intersection in Figure \ref{fig:pppzoom}.}
    \label{fig:+++BtoB}
\end{figure}

\begin{figure}[!h]
	\centering
        \includegraphics[width=0.55\textwidth]{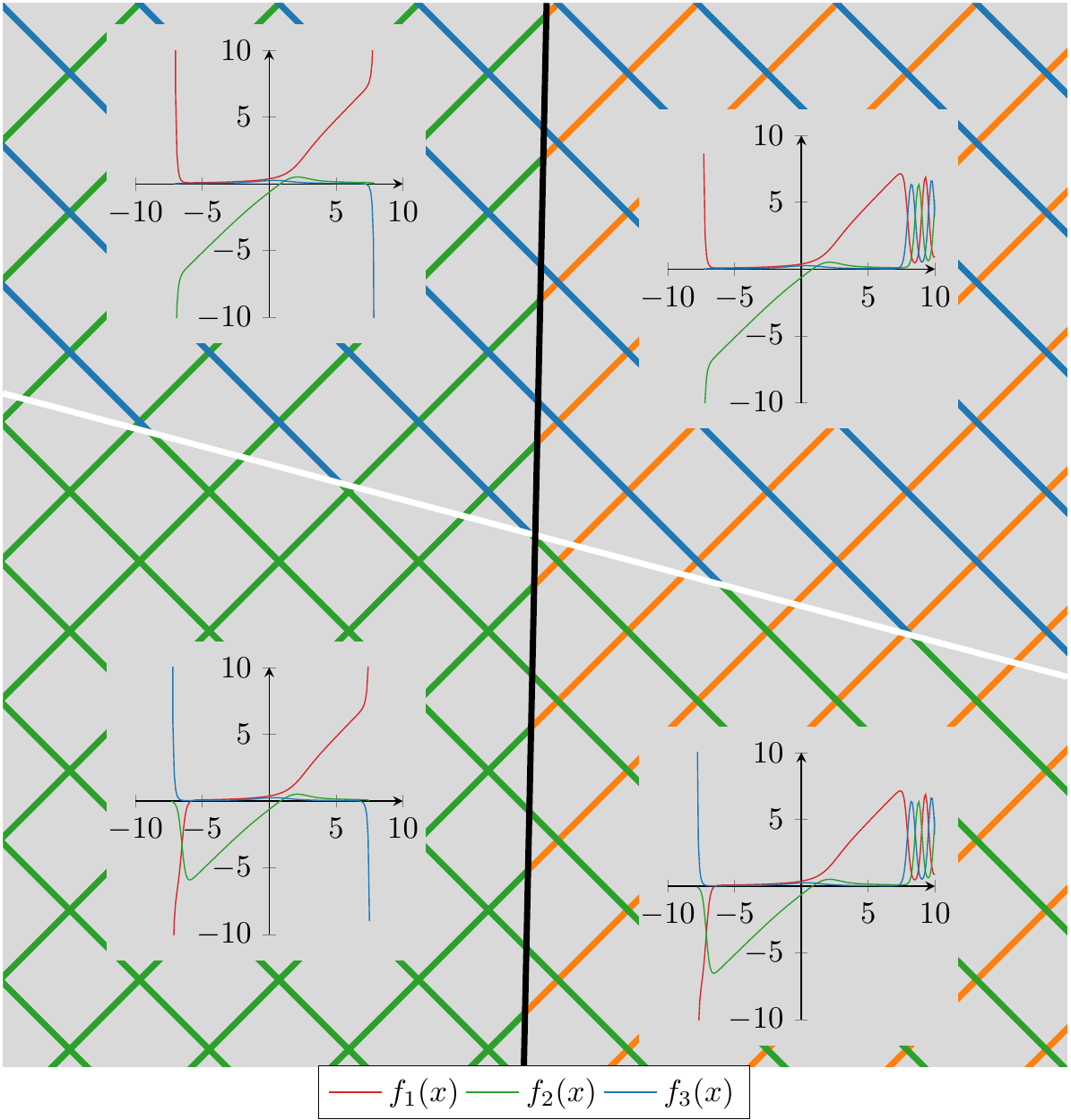}
		\caption{A zoomed-in version of a first-quadrant intersection in Figure \ref{fig:ppm}.}
		\label{fig:ppmzoom}
\end{figure}	

\begin{figure}[!h]
    \centering
    \includegraphics[width=0.55\textwidth]{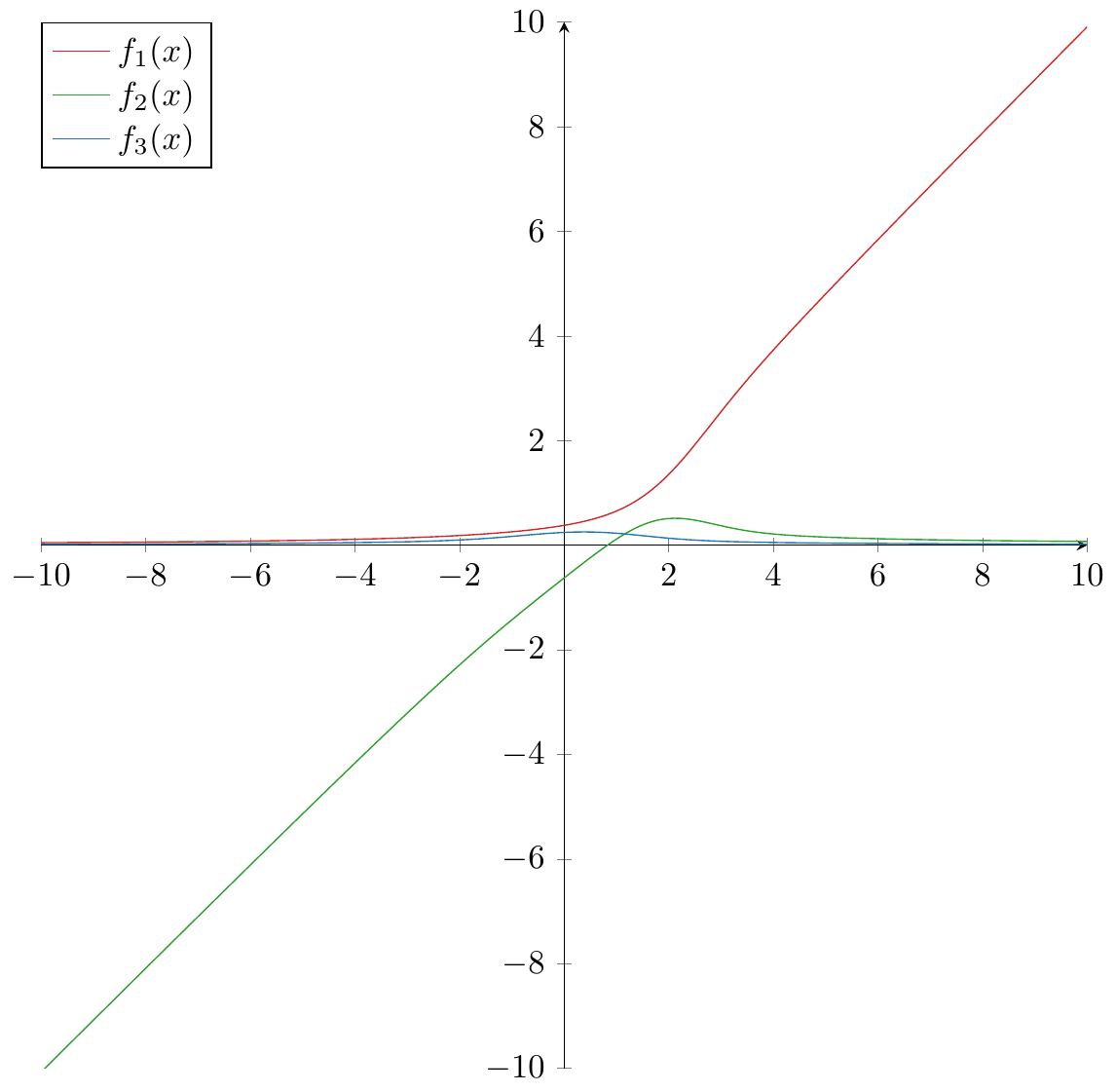}
    \caption{The $B_2^{-} \to B_1^{+}$ solution corresponding to the intersection in Figure \ref{fig:ppmzoom}.}
    \label{fig:++-BB1}
\end{figure}

\begin{figure}[!h]
	\centering
        \includegraphics[width=0.55\textwidth]{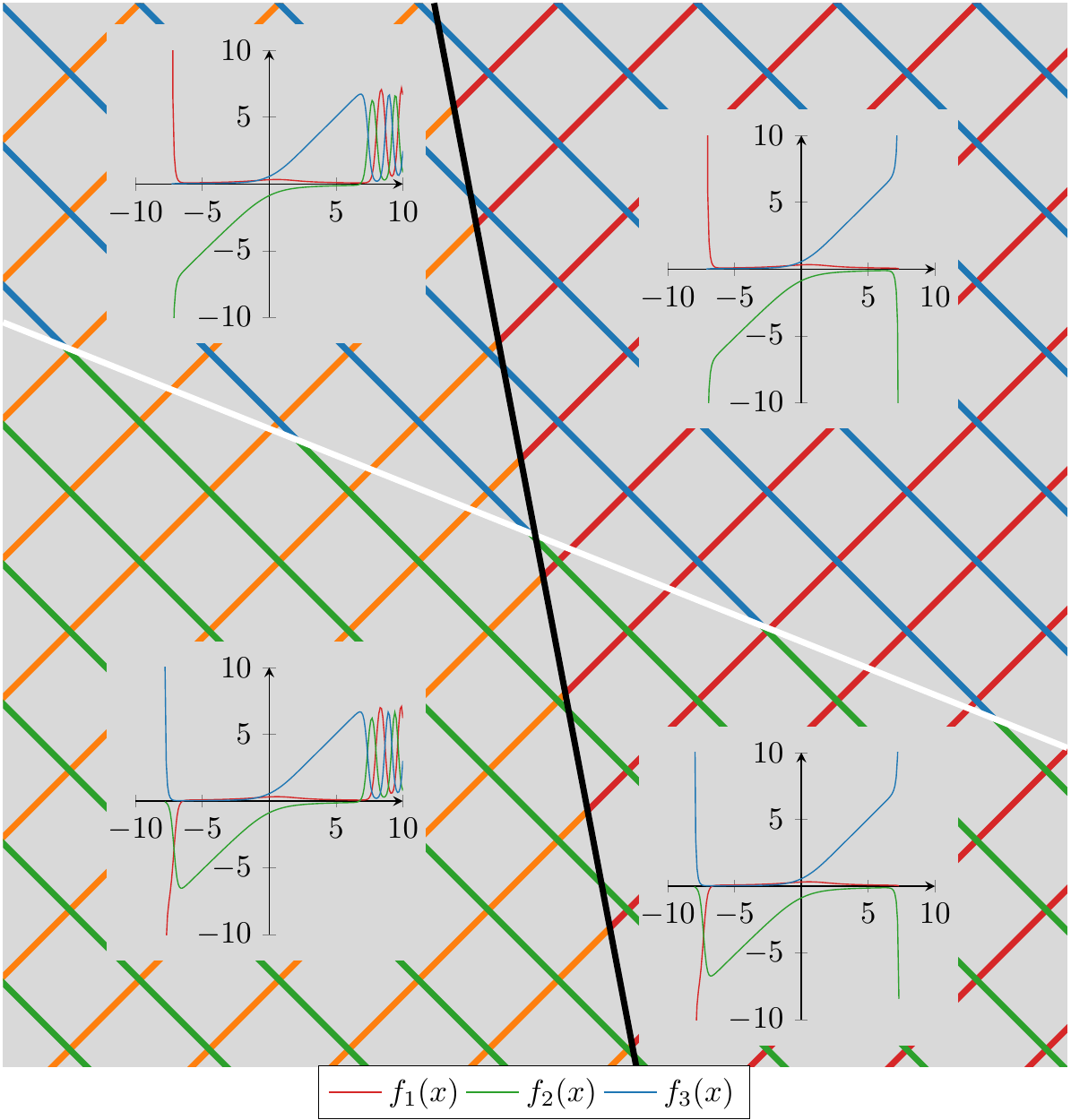}
		\caption{A zoomed-in version of the second first-quadrant intersection in Figure \ref{fig:ppm}.}
		\label{fig:ppmzoom2}
\end{figure}	

\begin{figure}[!h]
    \centering
    \includegraphics[width=0.55\textwidth]{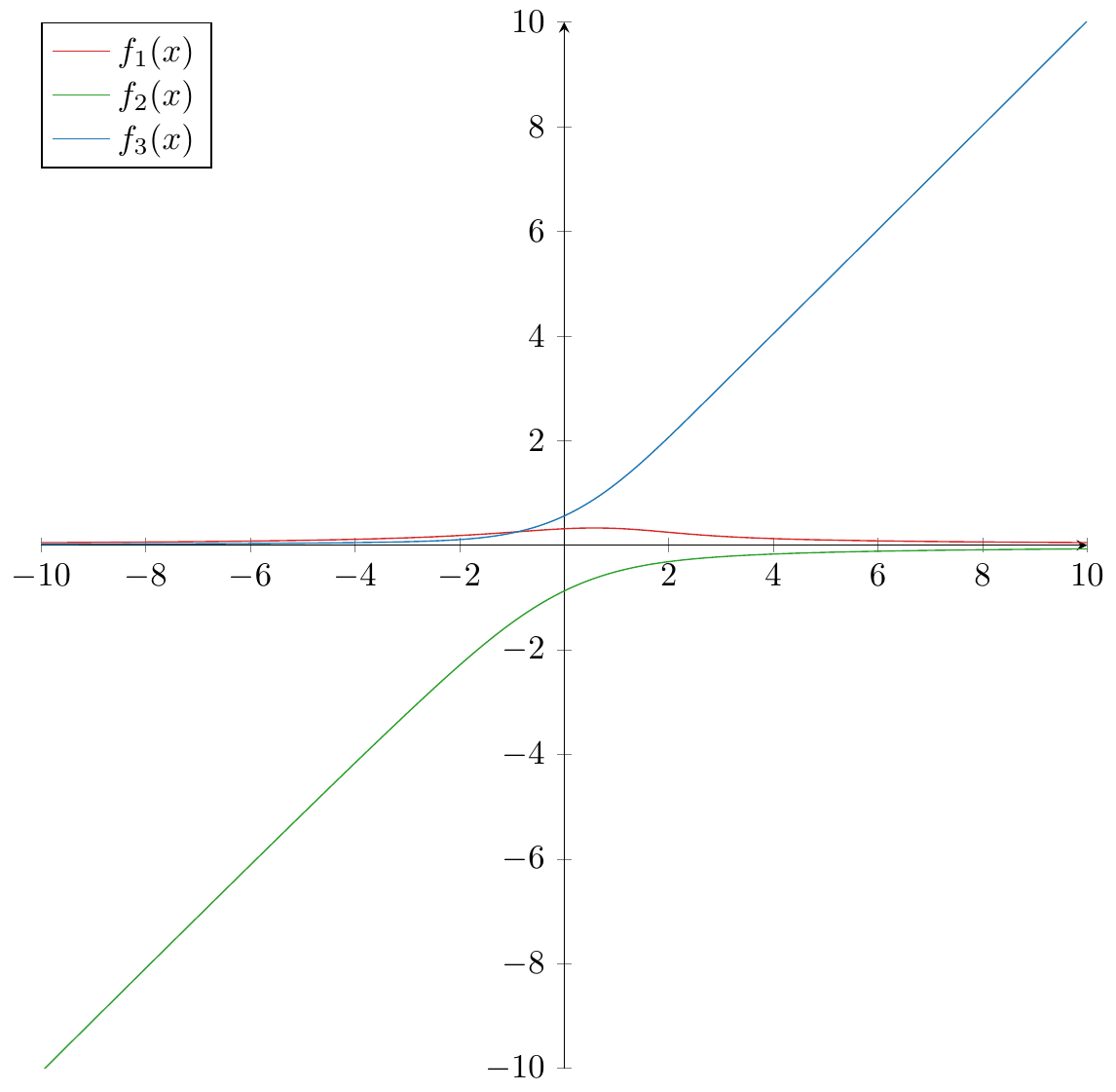}
    \caption{The $B_2^{-} \to B_3^{+}$ solution corresponding to the intersection in Figure \ref{fig:ppmzoom2}.}
    \label{fig:++-BB2}
\end{figure}


\begin{figure}[!h]
	\centering
	    \includegraphics[width=0.55\textwidth]{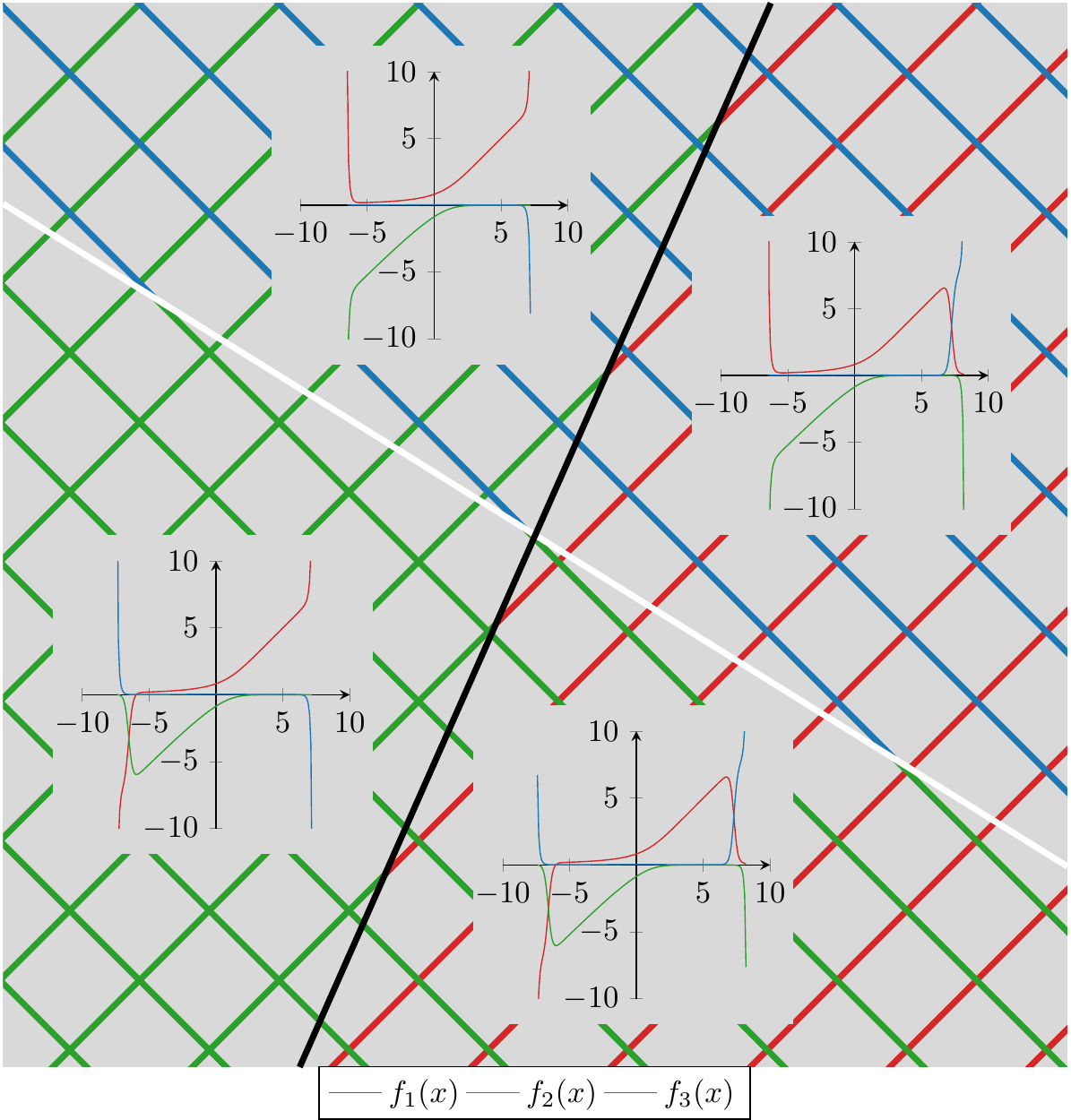}
		\caption{A zoomed-in version of the first-quadrant intersection in Figure \ref{fig:pmm}.}
		\label{fig:pmmzoom}
\end{figure}	

\begin{figure}[!h]
    \centering
    \includegraphics[width=0.55\textwidth]{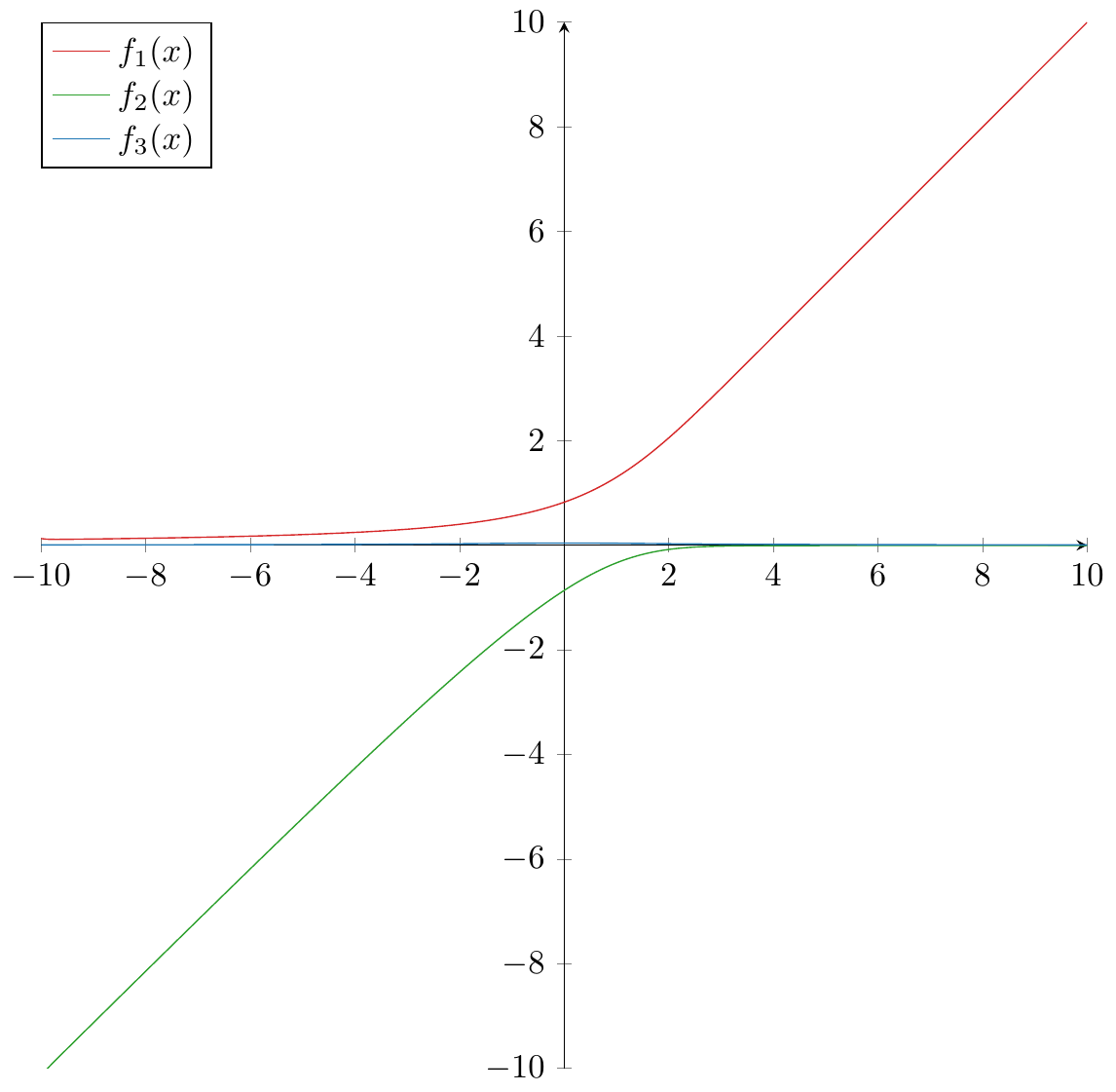}
    \caption{The $B_2^{-} \to B_1^{+}$ solution corresponding to the intersection in Figure \ref{fig:pmmzoom}.}
    \label{fig:+--BB1}
\end{figure}

\newpage



\end{document}